\documentclass[11pt]{article}
\usepackage[a4paper, margin=2.5cm]{geometry}

\usepackage{graphicx}
\usepackage{amsmath}
\usepackage{amssymb}
\usepackage{amsthm}
\usepackage{booktabs}
\usepackage[colorlinks=true, linkcolor=RoyalBlue, citecolor=RoyalBlue, urlcolor=RoyalBlue]{hyperref}
\usepackage[dvipsnames]{xcolor}
\usepackage{array}
\usepackage{natbib}
\usepackage{setspace}
\usepackage{subcaption}
\usepackage{placeins} 
\usepackage{rotating}
\usepackage{bbm}
\usepackage{tikz}
\usetikzlibrary{shapes.misc}

\definecolor{myblue}{HTML}{2C7FB8}
\definecolor{mygrey}{HTML}{6C757D}
\definecolor{myred}{HTML}{B2182B}

\graphicspath{{images/}}

\newcommand{\indep}{\mathrel{\perp \!\!\! \perp}}

\newtheorem{theorem}{Theorem}
\newtheorem{proposition}{Proposition}
\newtheorem{assumption}{Assumption}
\newtheorem{remark}{Remark}

\def\restriction#1#2{\mathchoice
              {\setbox1\hbox{${\displaystyle #1}_{\scriptstyle #2}$}
              \restrictionaux{#1}{#2}}
              {\setbox1\hbox{${\textstyle #1}_{\scriptstyle #2}$}
              \restrictionaux{#1}{#2}}
              {\setbox1\hbox{${\scriptstyle #1}_{\scriptscriptstyle #2}$}
              \restrictionaux{#1}{#2}}
              {\setbox1\hbox{${\scriptscriptstyle #1}_{\scriptscriptstyle #2}$}
              \restrictionaux{#1}{#2}}}
\def\restrictionaux#1#2{{#1\,\smash{\vrule height .8\ht1 depth .85\dp1}}_{\,#2}} 

\def\cC{\mathcal{C}}

\def\cE{\mathcal{E}}

\def\cG{\mathcal{G}}

\def\cK{\mathcal{K}}
\def\cL{\mathcal{L}}

\def\cX{\mathcal{X}}

\def\bX{\mathbf{X}}

\def\bz{\mathbf{z}}

\def\bbC{\mathbb{C}}

\def\bbE{\mathbb{E}}

\def\bbP{\mathbb{P}}

\def\bbR{\mathbb{R}}

\let\lac\{
\let\rac\}
\renewcommand{\{}{\left\lac}
\renewcommand{\}}{\right\rac}

\def\1{\mathbbm{1}}

\def\({\left(}
\def\){\right)}

\DeclareMathOperator{\supp}{supp}

\newcommand{\wt}{\widetilde}

\title{A Bayesian Approach to Causal Cure Models}
\author{
Emma Torrini$^{1}$, 
Maïlis Amico$^{2}$,  
Nicolas Molinari$^{2,3,4}$, 
Clément Berenfeld$^{3}$ \\ \\
{\small$^{1}$Department of Statistics, Computer Science, Applications, University of Florence, Florence, Italy} \\
{\small$^{2}$IDESP, University of Montpellier, INSERM, Montpellier, France} \\
{\small$^{3}$PreMeDICaL, Inria, INSERM, University of Montpellier, Montpellier, France}\\
{\small$^{4}$CHU Montpellier, Montpellier, France}
} 
\date{}

\begin{document}


\maketitle

\begin{abstract}
    Time-to-event data often include individuals who will never experience the failure event, and are therefore considered cured. In such settings, frequently encountered in clinical research, a substantial proportion of patients may remain event-free throughout the observation period, leading to the appearance of a survival plateau that is commonly interpreted as evidence of a cure fraction. Analyzing such data requires inference on the cure fraction and on the survival function for the uncured subpopulation, tasks which are traditionally achieved with mixture cure models. However, assessing the causal effect of a treatment on these quantities is non-trivial. We consider principal stratification causal estimands, which have been proposed to evaluate effects on the cure fraction and on the survival for an always-uncured stratum. We additionally introduce a novel estimand, considering the causal effect on the survival for a non-always-cured union of strata. The problem is framed from a Bayesian model-based perspective, that provides a flexible and unified estimation strategy while maintaining a direct link with classical mixture cure model quantities. We show how this framework can also accommodate additional information on the cure status, when available. The reliability of the proposed approach is validated through simulations, demonstrating competitive and robust performance relative to existing methods. Finally, we illustrate its practical usefulness through an application to a randomized trial comparing non-invasive ventilation with standard oxygen therapy in patients with hypoxemic respiratory failure following abdominal surgery. 
\end{abstract}


\clearpage

\section{Introduction}

\subsection{Motivation} Survival (or time-to-event) variables describe the time until an event occurs and play a central role in many scientific applications. Survival data require specialized statistical methods, as the time of interest is often only partially observed due to right-censoring. In causal settings, where the aim is to evaluate the effect of a treatment on an outcome, the same features complicate the definition and identification of treatment effects on time-to-event variables. A large body of work has therefore focused on adapting classical survival analysis tools to causal inference problems \citep[see][for a review]{voinot2025treatment}. 

Survival models commonly operate under the assumption that all individuals will eventually experience the event of interest, so that the survival function goes to zero in finite time. However, there are instances where a fraction of the units may never experience the failure event, even if followed for a long period of time. These subjects can be seen as being \emph{cured}, or long-term survivors, and can be thought of as having infinite failure times. In this instance, the survival function does not go to zero in finite time, but rather levels-off to a certain value, corresponding to the cure proportion. 

To illustrate this situation, we consider data from a randomized controlled trial conducted in France to evaluate whether non-invasive ventilation, compared with standard oxygen therapy, improves outcomes among patients with hypoxemic respiratory failure following abdominal surgery \citep[NIVAS study,][]{jaber2016nivas}. We focus on the time to first reintubation or death within the post-operative period. Because their health condition does not deteriorate, a fraction of patients may not require reintubation or experience death in the clinical horizon of interest, and can therefore be regarded as long-term survivors with respect to the failure event. 
In this context, the Kaplan-Meier estimator of the survival function does not approach zero within the study period; instead, a long and stable plateau is observed which, together with clinical evidence, suggests the presence of a cure fraction in the data (see e.g. Figure~\ref{fig:nivas_km} in Section~\ref{sec:nivas}). 

In such setting, two quantities become of interest alongside the overall survival function: the cure fraction and the survival function of the uncured individuals. However, because cured subjects have infinite failure times, they are all censored. It is therefore not possible to distinguish cured individuals from uncured and censored ones a priori, and tailored methods are required to achieve a separate identification of these two quantities. For this purpose, classical survival models have been extended to cure models. 
Cure models were originally introduced by \citet{boag1949maximum} and \citet{berkson1952survival} in the form of \emph{mixture cure models}, which assume that the overall survival function arises from a mixture between the survival function for cured and uncured subjects. A more recent class of cure models is \emph{promotion time cure models}, which were proposed by \citet{yakovlev1996stochastic}.
Many different methodological proposals have been made for these two classes of cure models in the literature \citep[see][for a detailed review]{amico2018cure}, but very little has been done when it comes to using cure models from a causal perspective. 

\subsection{Related work} When the goal is to assess the causal effect of a treatment on a survival outcome with a cure fraction, different inferential targets may be of interest. A first quantity concerns the effect on the cure proportion, which captures the treatment's ability to prevent subjects from experiencing the failure event. A second target is the effect on the overall survival, reflecting the combined impact of the treatment on both the cure probability and the timing of the event among uncured subjects. Because this marginal effect incorporates different mechanisms, one can further separate them, evaluating the causal effect on survival among uncured subjects, in order to determine whether the treatment is able to delay the event for those who remain susceptible. 

In this context, \citet{gao2017estimating} consider causal effect estimation in the presence of a cure fraction and noncompliance, proposing two estimands that target the treatment effect on the cure rate and on the survival for the uncured in the complier subgroup. In a related line of work, \citet{wang2023estimating} investigate causal inference in observational studies for survival data with a cure fraction, focusing on the first two inferential goals described above, namely the effect on the cure proportion and on the overall survival. 
While standard cure models are sufficient for the definition and estimation of these two effects, \citet{wang2024causal} argue that care must be paid in defining a causal estimand for the survival of the uncured individuals, since the set of uncured subjects under treatment might be different from that under control. 
Therefore, while it is possible to estimate the survival function of the uncured under the two treatment conditions, comparing those does not lead to valid causal conclusions, since they do not refer to a common subpopulation. The issue arises because the cure status is determined after treatment assignment and may be affected by it; conditioning on its observed value therefore amounts to selecting individuals based on a post-treatment quantity, inducing selection bias. This motivates the use of frameworks that explicitly handle post-treatment variables. With this idea in mind, \citet{wang2024causal} propose two causal estimands based on principal stratification \citep{frangakis2002principal}, targeting the effect on potential failure times in an ``always-uncured'' stratum, where individuals would not be cured regardless of the treatment received. For this task, they propose estimators based on mixture cure models. Specifically, they formulate a set of assumptions that link (i) the cure probability and the survival for the uncured obtained from mixture cure models, and (ii) principal strata probabilities and survival functions. More recently, and independently of our work, \citet{linero2026bayesian} proposed a Bayesian machine learning approach to causal cure models; however, they use different models and target different estimands than ours. 

Our contribution also collocates into the larger literature of Bayesian model-based principal stratification, first introduced by \citet{imbens1997bayesian}. While early work focused on continuous or binary endpoints, recent developments have extended this approach to time-to-event outcomes \citep{liu2024principal} and also addressed unique challenges such as survival intermediate variables \citep{ballerini2025evaluating, mattei2024assessing}.

\subsection{Contributions} We propose a model-based framework for causal inference in survival settings with a cure fraction. Following \citet{wang2024causal}, we adopt a principal stratification perspective and define four principal strata according to individuals’ potential cure status under treatment and control: always-cured, always-uncured, protected, and harmed. Within this framework, we target two main causal quantities: the treatment effect on the cure probability and the treatment effect on the survival among always-uncured individuals. In addition, we introduce a novel causal estimand based on the union of principal strata corresponding to the non-always-cured population. This quantity captures the treatment effect on the failure time in the overall susceptible population, net of the dilution induced by individuals who would be cured under both treatment conditions. We show that the cure probability and the survival function for uncured subjects, that characterize standard mixture cure models, can be expressed in terms of the causal quantities involved in our method. In particular, the likelihood underlying our approach is a natural decomposition of the classical mixture cure model likelihood, providing a direct link between this causal approach and conventional cure modeling. 
This formulation also accommodates situations where additional information on the cure status is available; for instance, in our motivating study, hospital discharge can imply a censored patient is cured from the post-operative outcomes of interest. We discuss how such information, typically overlooked in mixture cure models, can be incorporated into the likelihood to improve inferential precision.
Estimation is conducted within a Bayesian framework and implemented in Stan \citep{carpenter2017stan,stan}, yielding a flexible and efficient procedure that facilitates estimation and reproducibility.
A simulation study demonstrates similar performances of our method compared to that of \citet{wang2024causal}, with the additional benefit that our method is agnostic to the underlying assumptions (e.g., it does not require choosing a substitutional variable), making it easier to implement in practice. Moreover, the model-based formulation facilitates sensitivity analyses and the relaxation of key assumptions. 
We conclude by illustrating the practical usefulness of the proposed methodology with an application to the randomized study on non-invasive ventilation for hypoxemic acute respiratory failure introduced earlier.

\section{Causal cure models}
\label{sec:causcurmod}

Consider a sample with $n$ individuals, and let $Z_i$ be the treatment assigned to patient $i$, with $Z_i=1$ if the patient is assigned the investigational treatment, and $Z_i=0$ if the patient is assigned the control. We indicate with $z$ a realization of $Z_i$. Let $X_i$ be a vector of pre-treatment covariates, belonging to some covariate space $\cX$. We let $T_i$ be the time-to-event variable. Often this time can be censored by another variable $C_i$, and one only observes $Y_i = T_i \wedge C_i$. We let $\Delta_i = \mathbb{I}(T_i \leq C_i)$ be the event indicator, so that $\Delta_i = 1$ if the event is observed, and $\Delta_i = 0$ if not. In the presence of a cure fraction, $T_i \in \mathbb{R}^+ \cup \{\infty\}$, where $T_i=\infty$ indicates that the patient cannot experience the failure event in finite time, and therefore is cured. Consequently, we can define $B_i = \mathbb{I}(T_i<\infty)$ the uncured indicator, such as $B_i = 1$ if the subject is uncured, and $B_i = 0$ otherwise. Note that $B_i$ is only partially observed due to censoring. 
In such setting, the marginal survival function $S(t \mid X_i,Z_i) = \bbP(T_i >t \mid X_i,Z_i)$ can be defined as a mixture between the survival function for cured and uncured subjects \citep{farewell1982mixture}, such as
\begin{equation}
\label{eq:mixture} 
    S(t \mid X_i, Z_i) = 1 - p(X_i, Z_i) + p(X_i, Z_i)  S_{\mathrm{U}}(t \mid X_i, Z_i),
\end{equation}
where $p(X_i, Z_i) = \bbP(B_i=1 \mid X_i,Z_i)$ denotes the probability of being uncured, and $S_{\mathrm{U}}(t \mid X_i,Z_i) = \bbP(T_i > t \mid X_i,Z_i,B_i = 1)$ represents the survival function for uncured individuals. 

To define causal quantities, we adopt the potential outcomes framework \citep{rubin1974estimating}.  
Let $T_i(z)$ be the potential time to the outcome under treatment $z$, that is the outcome of patient $i$ if, possibly contrary to the fact, they had received treatment $z$, and let $B_i(z)$ be the potential status of being uncured under treatment $z$. Similarly, the potential censoring time is denoted as $C_i(z)$. 

We adopt the principal stratification framework, using the joint values of the potential uncured status under treatment and control to define four latent principal strata as in \citet{wang2024causal}:
\begin{itemize}
    \item Always-cured (CC): $\{i: B_i(1)=0, B_i(0)=0\}$, i.e., patients who would be cured regardless of treatment assignment; 
    \item Protected (CU): $\{i: B_i(1)=0, B_i(0)=1\}$, i.e., patients who would be cured under the investigational treatment but uncured under control;  
    \item Harmed (UC): $\{i: B_i(1)=1, B_i(0)=0\}$, i.e., patients who would be uncured under the investigational treatment but cured under control; 
    \item Always-uncured (UU): $\{i: B_i(1)=1, B_i(0)=1\}$, i.e., patients who would be uncured regardless of treatment assignment. 
\end{itemize}
We let $\cG := \{\mathrm{CC,CU,UC,UU}\}$ be the set of all strata and use $G_i \in \cG$ to indicate the principal stratum of patient $i$. It is worth noting that $B_i(1)$ and $B_i(0)$ are generally affected by the treatment, but the ordered pair $G_i=(B_i(1),B_i(0))$ is not. Therefore, any comparison of potential outcomes within a principal stratum is a well-defined causal effect, since one can condition on the principal stratum $G_i$ just as with any pre-treatment covariate. 

Patients in the CC stratum may receive either treatment, as neither option would lead to a failure event. In contrast, patients in the CU and UC strata benefit from only one of the two treatments (the investigational and the control treatment, respectively) and should therefore ideally receive the treatment that is effective for them. Patients in the UU stratum would not be cured by either treatment, and the preferred option is generally the treatment that maximizes survival. These considerations give rise to several inferential objectives. First, it is crucial to estimate principal stratum probabilities given the observed baseline covariates. Since the principal stratum is latent, characterizing strata using these covariates provides insight into the likely stratum membership of each patient, which can inform clinical decision-making. Second, comparing the survival between treatments within the UU stratum helps determining which option offers the best prognosis to these patients. Finally, comparisons of survival across unions of principal strata may also be of interest and provide additional clinically relevant summaries. 

\subsection{Causal estimands}
\label{sec:estimands}

Denote with $p^{(z)}(X_i) = \bbP(B_i(z)=1\mid X_i)$ the probability of being uncured under treatment $z\in\{0,1\}$, and let $\pi_{g}(X_i) = \bbP(G_i=g\mid X_i)$ be the probability of each stratum $g \in \cG$. Note that $p^{(1)}(X_i)=\pi_{\mathrm{UC}}(X_i)+\pi_{\mathrm{UU}}(X_i)$ and $p^{(0)}(X_i)=\pi_{\mathrm{CU}}(X_i)+\pi_{\mathrm{UU}}(X_i)$ (see Section~\ref{sec:bayes}). We will denote the marginal probability with an overline, so that $\bar p^{(z)} = \bbP(B_i(z)=1)$ and $\bar \pi_g = \bbP(G_i = g)$. One of the primary focuses in settings with a cured proportion is to evaluate the difference in cure rates, whose conditional expression is 
\begin{equation*}
    \delta(X_i)= (1-p^{(1)}(X_i))-(1-p^{(0)}(X_i)) = \pi_{\mathrm{CU}}(X_i) - \pi_{\mathrm{UC}}(X_i), 
\end{equation*}
and whose marginal counterpart writes
\begin{equation}
\label{eq:delta}
    \bar\delta = (1-\bar p^{(1)}) - (1-\bar p^{(0)}) = \bar\pi_{\mathrm{CU}} - \bar\pi_{\mathrm{UC}}.
\end{equation}

Another quantity of interest is a comparison between the survival of patients in a given stratum. We let, for any subset $K \subset \cG$,
$$
S_K^{(z)}(t) := \bbP(T_i(z) > t \mid G_i \in K).
$$
We consider two causal estimands for a given subset of strata $K \subset \cG$. The first is the principal survival causal effect (PSCE), which compares the survival probability of patients in the subset $K$ under the two treatment conditions at a fixed time $t\geq 0$: 
\begin{equation}
\label{eq:tau}
    \tau_{K}(t) = S_{K}^{(1)}(t) - S_{K}^{(0)}(t) .
\end{equation}
Second, for an overall comparison, we consider the principal difference in restricted mean survival time (RMST) \citep{royston2013restricted}, which compares the mean survival time of patients in the subset $K$ under the two treatment conditions until a pre-specified time $t^*>0$:
\begin{equation}
\label{eq:rmst}
    \tau_{K}^{{\rm RMST}}(t^*) = \int_0^{t^*} S_{K}^{(1)}(t) - S_{K}^{(0)}(t) \, \mathrm{d}t .
\end{equation}

In our principal stratification framework, the treatment has a non-trivial effect in the three strata UU, CU, and UC. This motivates causal estimands related to failure times in both the stratum UU (\emph{always-uncured} patients), and in the union of CU, UC, and UU strata, denoted by $\neg {\rm CC} := \{\rm CU, UC, UU\}$ (\emph{non-always-cured} patients).


\subsection{Assumptions and identification} \label{sec:identification}

Identification of principal causal effects is inherently challenging due to the latent nature of the strata. In our setting, this aspect is further amplified, as censoring generates situations in which no principal stratum can be excluded a priori.
Table~\ref{tab:profiles} shows what each observed profile reveals about a patient's principal stratum membership: when the outcome is censored, as in the first and third row, the patient could belong to any of the four strata. 
We now discuss the assumptions and identification strategies that we adopt. 

\begin{table}[h]
    \centering
    \caption{Principal stratum knowledge according to the observation.}
    \label{tab:profiles}
    \begin{tabular}{ccc}
        \hline
        Treatment $Z_i$ & Event status $\Delta_i$ & Possible stratum $G_i$ \\
        \hline
        $1$ & $0$ & $\rm CC$,$\rm CU$,$\rm UC$,$\rm UU$ \\
        $1$ & $1$ & $\rm UC$,$\rm UU$ \\
        $0$ & $0$ & $\rm CC$,$\rm CU$,$\rm UC$,$\rm UU$ \\
        $0$ & $1$ & $\rm CU$,$\rm UU$ \\
        \hline
    \end{tabular}
\end{table}

We observe, for each participant $i \in [n]$, a tuple $(Z_i, X_i, Y_i, \Delta_i)$ with latent potential outcomes $(T_i(1), T_i(0), C_i(1), C_i(0))$. We assume that the tuples $(Z_i, X_i, T_i(1), T_i(0), C_i(1), C_i(0))$ are i.i.d. and we require the following assumptions for all $i \in [n]$.

\begin{assumption}[SUTVA]  
\label{ass:sutva}
$T_i = T_i(z)$ and $C_i = C_i(z)$ when $Z_i=z$.    
\end{assumption}

\begin{assumption}[Ignorability]
\label{ass:ignorability}
$Z_i \indep \{T_i(1),T_i(0),C_i(1),C_i(0)\} \mid X_i$.
\end{assumption}

\begin{assumption}[Positivity of treatment]
\label{ass:positivity}
$0<\bbP(Z_i=1 \mid X_i)<1$.
\end{assumption}

\begin{assumption}[Conditionally independent censoring]
\label{ass:cond-indep-cens}
$C_i(z) \indep T_i(z) \mid X_i, Z_i=z$ for $z\in \{0,1\}$.
\end{assumption}

We let $t_{\max}^{(z)}(X_i)$ be the supremum of the support of the law of $T_i(z)\mid X_i, B_i(z)=1$. 

\begin{assumption}[Positivity of censoring]
\label{ass:suff-large-cens} 
$\bbP(C_i(z) > t \mid X_i) > 0$ for all time $t < t^* \wedge t_{\max}^{(z)}(X_i)$ and $z \in \{0,1\}$.
\end{assumption}

Assumption~\ref{ass:sutva} is a standard identifying assumption in causal inference, and postulates no interference between units and no different versions of the same treatment. 
Assumptions~\ref{ass:ignorability} and \ref{ass:positivity} are implied when the treatment is randomized. In general, they require that the treatment can be thought of as randomly assigned when conditioning on the covariates, and that every unit has a positive probability of receiving either treatment. Assumption~\ref{ass:cond-indep-cens}, which is commonly invoked in survival models, states that the potential censoring time is independent of the potential outcome time, conditional on the covariates and the treatment assignment. Assumption~\ref{ass:suff-large-cens} postulates that there is a positive probability that the survival function reaches the limiting level before censoring.

\begin{proposition} \label{prp:ident_1}
    Under Assumptions~\ref{ass:sutva}--\ref{ass:suff-large-cens}, the survival curves $S^{(z)}(t|X_i)$ and $S^{(z)}(t)$ are identifiable for all $t\leq t^*$. Furthermore, if $t^* \geq t_{\rm \max}^{(z)}(X_i)$ a.s. for both $z\in\{0,1\}$, then the difference in cure rates $\delta(X_i)$ and $\bar \delta$ are also identifiable.
\end{proposition}
The proof of this proposition relies on the identification formula
$$
S^{(z)}(t\mid X_i) = S(t\mid X_i,Z=z),
$$
which links the counterfactual survival function to the observable one $S(t\mid X_i,Z=z)$ defined in Equation~\eqref{eq:mixture}. It can be found in Appendix~\ref{app:proofs}. 
    
We now turn to the identification of quantities that are specific to strata. This requires stronger assumptions.

\begin{assumption}[Conditional monotonicity]
\label{ass:cond-mono}
It holds $\pi_{\rm UU}(X_i) = \min(p^{(1)}(X_i),p^{(0)}(X_i))$ almost surely.
\end{assumption}

Monotonicity assumptions are quite common in the principal stratification literature \citep{angrist1996identification}. Assumption~\ref{ass:cond-mono} is a particular type of monotonicity, implying that either $\pi_{\rm UC}(X_i)$ or $\pi_{\rm CU}(X_i)$ is zero for any given patient $i$. This assumption includes as a special case the classical monotonicity assumption, where researchers completely rule out the existence of one stratum, and it can be quite reasonable to assume in some contexts. For instance, in our motivating study, one might expect that a patient that is cured under control would be cured under treatment as well, so that $\pi_{\mathrm{UC}}(X_i)=0$ for all $i \in [n]$ almost surely. 

\begin{proposition} \label{prp:ident_2}
    We assume that $t^* \geq t_{\max}^{(z)}(X_i)$ a.s for $z \in \{0,1\}$. Then, under Assumptions~\ref{ass:sutva}--\ref{ass:cond-mono}, the principal strata proportions $\pi_g(X_i)$ and $\bar \pi_g$ for $g\in \cG$ are identifiable. 
\end{proposition}

The proof can be found in Appendix~\ref{app:proofs}. One needs one extra assumption to access the stratum-specific survival function $S_g^{(z)}(t|X_i)$. We propose two alternatives. The first one, called \emph{substitution relevance}, has been introduced by \citet{wang2024causal}. It amounts to postulate the existence of a subset $V_i$ of covariates that have an effect on the potential survival times $T_i(z)$ only through how they affect the stratum $G_i$. 

\begingroup

\renewcommand{\theassumption}{7.A}
\begin{assumption}[Substitution relevance]
\label{ass:vw}
There exists a latent decomposition of $X_i = (U_i,V_i)$ such that $V_i \indep T_i(z) \mid Z_i=z,G_i,U_i$ and $V_i \not\indep B_i(z) \mid U_i, B_i(1-z)=1$ for $z \in\{0,1\}$.
\end{assumption}
\begin{proposition} \label{prp:ident_3}
    We assume that $t^* \geq t_{\max}^{(z)}(X_i)$ a.s for $z \in \{0,1\}$. Then, under Assumptions~\ref{ass:sutva}--\ref{ass:cond-mono}, and \ref{ass:vw}, the survival functions $S^{(z)}_g(t|X_i)$ and $S^{(z)}_g(t)$ are identifiable for all $t \leq t^*$.
\end{proposition}

The proof already appears in \citet{wang2024causal}, but we provide it in Appendix~\ref{app:proofs}. In practice, it can be quite hard to find a substitutional variable $V_i$ or to even justify its existence. We propose therefore an alternative identification assumption, which posits a Cox model on the stratum-specific survival functions and a generic condition on the covariate support. Although we have formulated this assumption for the Cox model because of its popularity, the proof is based on a genericity argument and adapts to other general classes of models. 

\renewcommand{\theassumption}{7.B}
\begin{assumption}[Cox model]
\label{ass:cox} The survival functions $S_g^{(z)}(\cdot \mid X_i)$ for all applicable $g \in \neg {\rm CC}$ follow a Cox model. Furthemore, the support $\cX \subset \bbR^d$ of $X_i$ contains an open subset of $\bbR^d$. 
\end{assumption}
\endgroup

\begin{remark} The assumption on the support of $X_i$ implies that $X_i$ contains only continuous covariates. However, in many examples, including the one we illustrate in our simulation and case study, $X_i$ will contain some discrete covariate. We could adapt the previous assumption to ask that $X_i = (U_i,V_i)$ with $U_i$ continuous and $V_i$ discrete, and that there exists $v_0$ in the support of $V_i$ such that the support of $U_i$ given $V_i=v_0$ contains an open ball. The proofs would carry almost verbatim at the price of cumbersome notations.
\end{remark}

\begin{proposition}[Informal] \label{prp:ident_4}
    We assume that $t^* \geq t_{\max}^{(z)}(X_i)$ a.s. for $z \in \{0,1\}$ and that the strata probabilities $\pi_g(\cdot)$ for $g \in \cG$ are continuous. 
    Then, under Assumptions~\ref{ass:sutva}--\ref{ass:cond-mono}, and \ref{ass:cox}, the stratum-specific (conditional) survival functions are identifiable, unless $\supp \pi_g = \supp \pi_{\rm UU}$ and $\pi_g/\pi_{\rm UU}$ belongs to an exceptional set $\cE$, for $g\in\{\mathrm{CU,UC}\}$. 
    Furthermore, if the strata probabilities $\pi_g$ follow a multinomial-logistic model with intercept $a_g$ and slope $b_g$, then $\pi_g/\pi_{\rm UU}$ is not in $\cE$ as soon as $b_g \neq b_{\rm UU}$. 
\end{proposition}

A more formal version of this proposition along with its proof is provided in Appendix~\ref{app:proofs}. The set $\cE$ is exceptional in the sense that it is \emph{meagre}, a notion we recall in the Appendix. Meagre is a topological notion that roughly means that $\cE$ is too thin to occupy any substantive region of the space, so that $\pi_g/\pi_{\rm UU} \in \cE$ is unlikely to hold for a ratio of arbitrary continuous probability density. In the specific case where the strata probabilities are multinomial-logistic, which is the case considered in Section~\ref{sec:nivas}, the condition $\pi_g/\pi_{\rm UU} \notin \cE$ is automatically true if the slopes are distinct, so that identification is assured to hold.

The identification challenges discussed at the beginning of this section motivate the use of a Bayesian model-based approach, whose computations do not require full identification: as long as the prior is proper, the posterior remains well-defined even under partial identification \citep{lindley1972bayesian}. Within such framework, some of the identifying assumptions just outlined can be imposed or relaxed, facilitating sensitivity analyses.

\section{Bayesian principal stratification framework}  \label{sec:bayes}

\subsection{Likelihood and inference}

We now present a model-based approach that allows estimation of the causal effects introduced in Section~\ref{sec:estimands}. We start by deriving the likelihood of the mixture cure model in Equation~\eqref{eq:mixture} for a study with two groups, already introduced in \citet{wang2024causal}. For a time $t \geq 0$, define $S_{\mathrm{U}}^{(z)}(t \mid X_i) = \bbP(T_i(z)>t \mid X_i, B_i(z)=1)$ the conditional survival function of the potential outcome for uncured patients under treatment $z\in\{0,1\}$, and let $f_{\mathrm{U}}^{(z)}(t \mid X_i)$ be the associated probability density function. Recall that $p^{(z)}(X_i)=\bbP(B_i(z)=1 \mid X_i)$. 

\begin{theorem}
\label{thm:likelihood-standard}
    Under Assumptions~\ref{ass:sutva}, \ref{ass:ignorability}, \ref{ass:cond-indep-cens}, the likelihood function is 
    \begin{equation*}
    \begin{split}
        \prod_{z\in\{0,1\}} \prod_{i: Z_i=z} 
        & \big[ 1-p^{(z)}(X_i) + p^{(z)}(X_i) S_{\mathrm{U}}^{(z)}(Y_i \mid X_i) \big] ^{(1-\Delta_i)} \\
        &\times \big[ p^{(z)}(X_i) f_{\mathrm{U}}^{(z)}(Y_i \mid X_i) \big] ^{\Delta_i} ,
    \end{split} 
    \end{equation*}
    which is the product of two mixture cure model likelihoods under the two treatment conditions.
\end{theorem}

The proof of Theorem~\ref{thm:likelihood-standard} is provided in Appendix~\ref{app:proofs}. A direct link exists between the mixture cure model components in Theorem~\ref{thm:likelihood-standard} and the principal strata probabilities and survival functions, which constitute the objective of our inference. Under the investigational treatment, the set of uncured patients corresponds to the union of the UC and UU strata. Therefore, the probability of being uncured under treatment is
\begin{equation}
\label{eq:peq1}
    p^{(1)}(X_i)= \pi_{\mathrm{UC}}(X_i) + \pi_{\mathrm{UU}}(X_i) ,
\end{equation}
and the survival function for uncured patients under treatment is 
\begin{equation}
\label{eq:Seq1}
    S_{\rm U}^{(1)}(t\mid X_i)= 
    \frac{\pi_{\mathrm{UC}}(X_i) S_{\mathrm{UC}}^{(1)}(t\mid X_i) + \pi_{\mathrm{UU}}(X_i) S_{\mathrm{UU}}^{(1)}(t\mid X_i)}{\pi_{\mathrm{UC}}(X_i) + \pi_{\mathrm{UU}}(X_i)}.
\end{equation}
Similarly, under control, the set of uncured patients is the union of the CU and UU strata, so that the probability of being cured under control is 
\begin{equation}
\label{eq:peq2}
    p^{(0)}(X_i)= \pi_{\mathrm{CU}}(X_i) + \pi_{\mathrm{UU}}(X_i) ,
\end{equation}
and the survival probability for uncured patients under control is 
\begin{equation}
\label{eq:Seq2}
    S_{\rm U}^{(0)}(t\mid X_i)= 
    \frac{\pi_{\mathrm{CU}} (X_i)S_{\mathrm{CU}}^{(0)}(t\mid X_i) + \pi_{\mathrm{UU}}(X_i) S_{\mathrm{UU}}^{(0)}(t\mid X_i)}{\pi_{\mathrm{CU}}(X_i) + \pi_{\mathrm{UU}}(X_i)}.
\end{equation}
These equalities are proven in Appendix~\ref{app:proofs}. 

We are now in a position to derive the causal cure model likelihood, which will serve as the base for our Bayesian principal stratification inference. We let the strata probabilities $\pi_g(X_i)$ depend on a vector of parameters $\theta_G$, and write $\pi_g(X_i \mid \theta_G)$ for the parameterization. For a time $t \geq 0$, let $f_g^{(z)}(t \mid X_i)$ be the probability density function associated to the counterfactual survival function in stratum $g$. Note that this function is well-defined only for the principal strata UC and UU under treatment $z=1$, and CU and UU under treatment $z=0$. We let the functions $f_g^{(z)}$ and $S_g^{(z)}$ depend on a vector of parameters $\theta_g^{(z)}$, writing $S_g^{(z)}(t \mid X_i, \theta_g^{(z)})$ and $f_g^{(z)}(t \mid X_i, \theta_g^{(z)})$. Finally, let $\theta=\{\theta_G,\theta_{\mathrm{UC}}^{(1)},\theta_{\mathrm{UU}}^{(1)},\theta_{\mathrm{CU}}^{(0)},\theta_{\mathrm{UU}}^{(0)}\}$ and $O_n = (Z_i,X_i,Y_i,\Delta_i)_{i=1}^n$. 

Based on Theorem~\ref{thm:likelihood-standard} and Equations~\eqref{eq:peq1}--\eqref{eq:Seq2}, we can get an explicit formulation of the causal cure model likelihood: 
{\allowdisplaybreaks
\begin{align*}
    \mathcal{L}_n(\theta \mid O_n) \propto 
    &\prod_{i: Z_i=1} 
    \big[ \pi_{\mathrm{CC}}(X_i\mid\theta_G) + \pi_{\mathrm{CU}}(X_i\mid\theta_G)  \\
    &\quad +\pi_{\mathrm{UC}}(X_i\mid\theta_G) S_{\mathrm{UC}}^{(1)}(Y_i \mid X_i,\theta_{\mathrm{UC}}^{(1)}) + \pi_{\mathrm{UU}}(X_i\mid\theta_G) S_{\mathrm{UU}}^{(1)}(Y_i \mid X_i,\theta_{\mathrm{UU}}^{(1)}) \big] ^{(1-\Delta_i)} \\ 
    &\quad \times \big[ \pi_{\mathrm{UC}}(X_i\mid\theta_G) f_{\mathrm{UC}}^{(1)}(Y_i \mid X_i,\theta_{\mathrm{UC}}^{(1)}) + \pi_{\mathrm{UU}}(X_i\mid\theta_G) f_{\mathrm{UU}}^{(1)}(Y_i \mid X_i,\theta_{\mathrm{UU}}^{(1)}) \big] ^{\Delta_i} \\
    &\times \prod_{i: Z_i=0} \big[ \pi_{\mathrm{CC}}(X_i\mid\theta_G) + \pi_{\mathrm{UC}}(X_i\mid\theta_G) \\
    &\quad +\pi_{\mathrm{CU}}(X_i\mid\theta_G) S_{\mathrm{CU}}^{(0)}(Y_i \mid X_i,\theta_{\mathrm{CU}}^{(0)}) + \pi_{\mathrm{UU}}(X_i\mid\theta_G) S_{\mathrm{UU}}^{(0)}(Y_i \mid X_i,\theta_{\mathrm{UU}}^{(0)}) \big] ^{(1-\Delta_i)} \\ 
    &\quad \times \big[ \pi_{\mathrm{CU}}(X_i\mid\theta_G) f_{\mathrm{CU}}^{(0)}(Y_i \mid X_i,\theta_{\mathrm{CU}}^{(0)}) + \pi_{\mathrm{UU}}(X_i\mid\theta_G) f_{\mathrm{UU}}^{(0)}(Y_i \mid X_i,\theta_{\mathrm{UU}}^{(0)}) \big] ^{\Delta_i} .
\end{align*} 
}
This likelihood is induced by Table~\ref{tab:profiles} and can be seen as a decomposition of the likelihood in Theorem~\ref{thm:likelihood-standard}. Its form implies that we must specify two models: one for the principal stratum given covariates, and one for the potential outcomes given the principal stratum and covariates. 

The posterior distribution we aim to draw from is $\Pi_n(\theta \mid O_n) \propto \mathcal{L}_n(\theta \mid O_n) \cdot \Pi(\theta)$.
After specifying suitable prior distributions for each parameter in $\theta$, posterior sampling can be carried out in several ways. 
For ease of implementation and reproducibility, we rely, in the next sections, on the Stan programming language, which avoids analytical derivations and performs posterior sampling using Hamiltonian Monte Carlo. 

Once posterior inference on the parameters in $\theta$ is derived, the posterior distribution of the estimands can be obtained by exploiting the analytical relationship between the quantities of interest and the model parameters. For the estimand in Equation~\eqref{eq:delta}, we have that $\bar\pi_g = \mathbb{E}[\pi_g(X_i)]$.
The estimands in Equations~\eqref{eq:tau}--\eqref{eq:rmst} rely on the marginal stratum-specific survival functions under the two treatment statuses. As noted in \citet{liu2024principal}, we have that 
\begin{equation}
\label{eq:Sg}
    S_g^{(z)}(t) = \mathbb{E} \left[ \frac{\pi_g(X_i)}{\bar\pi_g} S_g^{(z)}(t\mid X_i) \right]
\end{equation}
for $z\in\{0,1\}$ and $g \in \cG$. 
In general, for a given subset $K \subset \cG$, we can show that
\begin{equation}
\label{eq:Sunion}
S_{K}^{(z)}(t) = \mathbb{E}\left[ \frac{\sum_{g\in K} \pi_g(X_i) S_g^{(z)}(t \mid X_i)}{\sum_{g\in K} \bar\pi_g} \right] .
\end{equation}
We note that $S_g^{(z)}(t \mid X_i) = 1$ for $z=1$, $g\in\{\mathrm{CC},\mathrm{CU}\}$ and for $z=0$, $g\in\{\mathrm{CC},\mathrm{UC}\}$, for all $t \geq 0$. 
Equations~\eqref{eq:Sg} and \eqref{eq:Sunion} are proven in Appendix~\ref{app:proofs}. 
At each MCMC iteration, we use these relationships to marginalize over the covariates by replacing each parameter-dependent quantity with its current posterior draw and substituting expected values with empirical averages.


\subsection{Introducing auxiliary information}
Sometimes, additional information on the patient's cure status may be available. In our motivating study, a patient who is discharged from the hospital without experiencing the outcome of interest can be considered cured. Other examples include patients declared disease-free after a clinical assessment, or the absence of disease recurrence after a sufficiently long observation period.  
In a standard mixture cure model, as the one in Theorem~\ref{thm:likelihood-standard}, all censored patients have an unknown cure status, and therefore contribute to the likelihood with a mixture between the cured and uncured survival function. This includes censored patients who are discharged from the hospital or, in general, who are known to be cured. In this section, we view hospital discharge as a censoring time carrying information on the patient's cure status, and show how this information can be included in the likelihood. 

Let $C_i^{D} \in \mathbb{R}^+ \cup \{\infty\}$ be the time to discharge for patient $i$, with $C_i^{D}\in\mathbb{R}^+$ if $B_i=0$ and $C_i^{D}=\infty$ if $B_i=1$. Note that here the hospital discharge time is well defined only for cured patients, implicitly assuming perfect classification: every discharged patient is cured and no uncured patient is mistakenly discharged. In terms of the outcome survival time, this means that either $T_i<\infty$ and $C_i^D=\infty$, or $T_i=\infty$ and $C_i^D<\infty$. We also introduce a binary indicator of hospital discharge, which we define as $D_i=\mathbb{I}(C_i^D \leq C_i)$. This indicator, in combination with the event indicator $\Delta_i$, creates three mutually exclusive categories:
\begin{enumerate}
    \item $(\Delta_i=0,D_i=0)$, when the patient is censored at $C_i$ and the cure status remains unknown (e.g., administrative censoring, loss to follow-up);
    \item $(\Delta_i=0,D_i=1)$, when the patient is censored at $C_i^D$ and thus identified as cured (e.g., censoring due to hospital discharge);
    \item $\Delta_i=1$, when the patient experiences the event during follow-up, and is therefore uncured.
\end{enumerate}
The construction of these three categories is clarified in Appendix~\ref{app:threecat}. We also note that the observed time $Y_i$ now modifies as $Y_i=\min(T_i,C_i,C_i^D)$. Let $C^D_i(z)$ be the potential hospital discharge time under treatment $z\in\{0,1\}$. For a time $t\geq 0$, define $S_{\neg\mathrm{U}}^{C^D(z)}(t \mid X_i) = \mathbb{P}(C^D_i(z)>t \mid X_i, B_i(z)=0)$ the potential discharge survival function for cured patients under treatment $z$, and let $f_{\neg\mathrm{U}}^{C^D(z)}(t \mid X_i)$ be the associated probability density function. 

Before presenting the likelihood function with the additional information, we clarify that relevant assumptions are now assumed to hold for $C^D_i$ as well: 
(i) Assumption~\ref{ass:sutva} also includes $C^D_i=C^D_i(z)$ when $Z_i=z$; 
(ii) Assumption~\ref{ass:ignorability} modifies as $Z_i \indep \{T_i(1),T_i(0),C_i(1),C_i(0),C^D_i(1),C^D_i(0)\} \mid X_i$; 
(iii) Assumption~\ref{ass:cond-indep-cens} modifies as $C_i(z) \indep \{T_i(z), C^D_i(z)\} \mid X_i, Z_i=z$ for $z\in\{0,1\}$. In the next proposition and the related Appendix, any reference to Assumptions~\ref{ass:sutva}, \ref{ass:ignorability}, and \ref{ass:cond-indep-cens} is intended as a reference to these expanded versions.

\begin{proposition}
\label{prp:likelihood-standard-cure}
    Under Assumptions~\ref{ass:sutva}, \ref{ass:ignorability}, \ref{ass:cond-indep-cens}, the likelihood function of Theorem~\ref{thm:likelihood-standard} modifies as follows to include information on the cure status for censored patients: 
    \begin{align*}
        \prod_{z\in\{0,1\}} \prod_{i:Z_i=z} 
        &\big[ \{1-p^{(z)}(X_i)\} S_{\neg\mathrm{U}}^{C^D(z)}(Y_i \mid X_i) \\
        &\quad + p^{(z)}(X_i) S^{(z)}_{\mathrm{U}}(Y_i \mid X_i) \big] ^{(1-\Delta_i)(1-D_i)} \\
        &\times  
        \left[ \{1-p^{(z)}(X_i)\} f_{\neg\mathrm{U}}^{C^D(z)}(Y_i \mid X_i) \right] ^{(1-\Delta_i)D_i} \\
        &\times  
        \left[ p^{(z)}(X_i) f^{(z)}_{\mathrm{U}}(Y_i \mid X_i) \right] ^{\Delta_i} .
    \end{align*}
\end{proposition}

The proof of Proposition~\ref{prp:likelihood-standard-cure} is provided in Appendix~\ref{app:proofs}.
In addition to Equations~\eqref{eq:peq1}--\eqref{eq:Seq2}, we also find a link between $S_{\neg\mathrm{U}}^{C^D(z)}(t \mid X_i)$ and $S_g^{C^D(z)}(t \mid X_i):= \mathbb{P}(C_i^D(z)>t \mid X_i,G_i=g)$ for any applicable $g\in\mathcal{G}$: 
\begin{equation}
\label{eq:SDeq1}
    S_{\neg\mathrm{U}}^{C^D(1)}(t\mid X_i)= 
    \frac{\pi_{\mathrm{CC}}(X_i) S_{\mathrm{CC}}^{C^D(1)}(t\mid X_i) + \pi_{\mathrm{CU}}(X_i) S_{\mathrm{CU}}^{C^D(1)}(t\mid X_i)}{\pi_{\mathrm{CC}}(X_i) + \pi_{\mathrm{CU}}(X_i)},
\end{equation}
\begin{equation}
\label{eq:SDeq2}
    S_{\neg\mathrm{U}}^{C^D(0)}(t\mid X_i)= 
    \frac{\pi_{\mathrm{CC}}(X_i) S_{\mathrm{CC}}^{C^D(0)}(t\mid X_i) + \pi_{\mathrm{UC}}(X_i) S_{\mathrm{UC}}^{C^D(0)}(t\mid X_i)}{\pi_{\mathrm{CC}}(X_i) + \pi_{\mathrm{UC}}(X_i)}.
\end{equation}
The proof is equivalent to that of Equations~\eqref{eq:Seq1} and \eqref{eq:Seq2}. 

We can now derive the causal cure model likelihood with auxiliary information. In addition to the parameterizations introduced earlier, we also let the functions $f_g^{C^D(z)}$ and $S_g^{C^D(z)}$ depend on a vector of parameters $\theta_g^{C^D(z)}$, writing $S_g^{C^D(z)}(t \mid X_i, \theta_g^{C^D(z)})$ and $f_g^{C^D(z)}(t \mid X_i, \theta_g^{C^D(z)})$. We now have $\theta^+=\{\theta_G,\theta_{\mathrm{UC}}^{(1)},\theta_{\mathrm{UU}}^{(1)},\theta_{\mathrm{CU}}^{(0)},\theta_{\mathrm{UU}}^{(0)}, \theta_{\mathrm{CC}}^{C^D(1)},\theta_{\mathrm{CU}}^{C^D(1)},\theta_{\mathrm{CC}}^{C^D(0)},\theta_{\mathrm{UC}}^{C^D(0)}\}$ and $O_n^+ = (Z_i,X_i,Y_i,\Delta_i,D_i)_{i=1}^n$. 

Using Proposition~\ref{prp:likelihood-standard-cure} and Equations~\eqref{eq:peq1}--\eqref{eq:Seq2}, \eqref{eq:SDeq1}--\eqref{eq:SDeq2} allows expressing the causal cure model likelihood with additional information on the cure status:
{\allowdisplaybreaks
\begin{align*}
    &\mathcal{L}_n^+(\theta^+ \mid O_n^+) \propto \\
    &\quad \prod_{i: Z_i=1} 
    \big[ \pi_{\mathrm{CC}}(X_i\mid\theta_G) S_{\mathrm{CC}}^{C^D(1)}(Y_i \mid X_i,\theta_{\mathrm{CC}}^{C^D(1)}) + \pi_{\mathrm{CU}}(X_i\mid\theta_G) S_{\mathrm{CU}}^{C^D(1)}(Y_i \mid X_i,\theta_{\mathrm{CU}}^{C^D(1)}) \\
    &\qquad +\pi_{\mathrm{UC}}(X_i\mid\theta_G) S_{\mathrm{UC}}^{(1)}(Y_i \mid X_i,\theta_{\mathrm{UC}}^{(1)}) + \pi_{\mathrm{UU}}(X_i\mid\theta_G) S_{\mathrm{UU}}^{(1)}(Y_i \mid X_i,\theta_{\mathrm{UU}}^{(1)}) \big] ^{(1-\Delta_i)(1-D_i)} \\
    &\qquad \times \big[ \pi_{\mathrm{CC}}(X_i\mid\theta_G) f_{\mathrm{CC}}^{C^D(1)}(Y_i \mid X_i,\theta_{\mathrm{CC}}^{C^D(1)}) + \pi_{\mathrm{CU}}(X_i\mid\theta_G) f_{\mathrm{CU}}^{C^D(1)}(Y_i \mid X_i,\theta_{\mathrm{CU}}^{C^D(1)}) \big] ^{(1-\Delta_i)D_i} \\
    &\qquad \times \big[ \pi_{\mathrm{UC}}(X_i\mid\theta_G) f_{\mathrm{UC}}^{(1)}(Y_i \mid X_i,\theta_{\mathrm{UC}}^{(1)}) + \pi_{\mathrm{UU}}(X_i\mid\theta_G) f_{\mathrm{UU}}^{(1)}(Y_i \mid X_i,\theta_{\mathrm{UU}}^{(1)}) \big] ^{\Delta_i} \\
    &\quad \times \prod_{i: Z_i=0} \big[ \pi_{\mathrm{CC}}(X_i\mid\theta_G) S_{\mathrm{CC}}^{C^D(0)}(Y_i \mid X_i,\theta_{\mathrm{CC}}^{C^D(0)}) + \pi_{\mathrm{UC}}(X_i\mid\theta_G) S_{\mathrm{UC}}^{C^D(0)}(Y_i \mid X_i,\theta_{\mathrm{UC}}^{C^D(0)}) \\
    &\qquad +\pi_{\mathrm{CU}}(X_i\mid\theta_G) S_{\mathrm{CU}}^{(0)}(Y_i \mid X_i,\theta_{\mathrm{CU}}^{(0)}) + \pi_{\mathrm{UU}}(X_i\mid\theta_G) S_{\mathrm{UU}}^{(0)}(Y_i \mid X_i,\theta_{\mathrm{UU}}^{(0)}) \big] ^{(1-\Delta_i)(1-D_i)} \\
    &\qquad \times \big[ \pi_{\mathrm{CC}}(X_i\mid\theta_G) f_{\mathrm{CC}}^{C^D(0)}(Y_i \mid X_i,\theta_{\mathrm{CC}}^{C^D(0)}) + \pi_{\mathrm{UC}}(X_i\mid\theta_G) f_{\mathrm{UC}}^{C^D(0)}(Y_i \mid X_i,\theta_{\mathrm{UC}}^{C^D(0)}) \big] ^{(1-\Delta_i) D_i} \\
    &\qquad \times \big[ \pi_{\mathrm{CU}}(X_i\mid\theta_G) f_{\mathrm{CU}}^{(0)}(Y_i \mid X_i,\theta_{\mathrm{CU}}^{(0)}) + \pi_{\mathrm{UU}}(X_i\mid\theta_G) f_{\mathrm{UU}}^{(0)}(Y_i \mid X_i,\theta_{\mathrm{UU}}^{(0)}) \big] ^{\Delta_i} 
    .
\end{align*} 
}
This likelihood requires specifying a model for the potential discharge time given the principal stratum and covariates, in addition to the two models for the principal stratum and potential outcomes previously mentioned.

\section{Simulation study}

This simulation study aims at comparing the performance of our model with that of \citet{wang2024causal} under different conditions.

\subsection{Modeling and prior details}
\label{sec:models-priors-simulation}

We consider a data generating process (DGP) respecting Assumption~\ref{ass:cond-mono}, thus replicating the model specification in \citet{wang2024causal}. To ensure a fair comparison with the estimation strategy proposed in \citet{wang2024causal} and coherence with their DGP, we also model the uncured probabilities with a logistic regression. We let $\bX_i := (1,X_i)$ and $\theta_G := \{\alpha^{(z)}\}_{z\in\{0,1\}}$ be two vectors of the same size as $\bX_i$. Letting ${\rm expit}(a):=\exp(a)/(1+\exp(a))$, we consider the parameterization 
$p^{(z)}(X_i \mid \alpha^{(z)}) :={\rm expit}(\bX_i^\top\alpha^{(z)})$ for  $z\in\{0,1\}$.
We then fix a parameter $\rho \in (-1,1)$ and derive the principal strata probabilities as: 
{\allowdisplaybreaks
\begin{align*}
\pi_{\mathrm{UU}}(X_i\mid \theta_G) &= \rho \min(p^{(1)}(X_i \mid \alpha^{(1)}), p^{(0)}(X_i \mid \alpha^{(0)})) \\
&\quad + (1-\rho) p^{(1)}(X_i \mid \alpha^{(1)}) p^{(0)}(X_i \mid \alpha^{(0)}); \\
\pi_{\mathrm{UC}}(X_i\mid \theta_G) &= p^{(1)}(X_i \mid \alpha^{(1)}) - \pi_{\mathrm{UU}}(X_i\mid \theta_G);\\
\pi_{\mathrm{CU}}(X_i\mid \theta_G) &= p^{(0)}(X_i \mid \alpha^{(0)}) - \pi_{\mathrm{UU}}(X_i\mid \theta_G);\\
\pi_{\mathrm{CC}}(X_i\mid \theta_G) &= 1 - \pi_{\mathrm{CU}}(X_i\mid \theta_G) - \pi_{\mathrm{UC}}(X_i\mid \theta_G) - \pi_{\mathrm{UU}}(X_i\mid \theta_G). 
\end{align*}
}
As discussed in \citet{wang2024causal}, $\rho$ represents the partial correlation between $B_i(1)$ and $B_i(0)$ given $X_i$, and it holds that $\rho=1$ under Assumption~\ref{ass:cond-mono}. This parameterization is chosen for the purposes of this simulation study, although we note that these equations do not constrain the principal stratum probabilities to the probability simplex for all values of $\rho$. For this reason, an alternative specification will be used in the data analysis (see Section~\ref{sec:nivas-models}). 

Many different models can be considered for the potential survival outcomes, ranging from fully parametric to semi-parametric ones. We choose a piecewise constant hazard model, sometimes referred to as the piecewise exponential model. We partition the time axis in $J$ intervals $(0,s_1], (s_1,s_2], \dots, (s_{J-1},s_J]$, with $0<s_1<\dots<s_J$ and $s_J$ greater than the maximum failure or censoring time observed for the outcome, so that $s_J>Y_i$ for all $i\in[n]$. In each interval $j$, for $z\in\{0,1\}$ and each applicable $g \in \neg {\rm CC}$, we assume a constant baseline hazard $\lambda_{gj}^{(z)}$, and we model the hazard function as 
\begin{equation*}
    h_g^{(z)}(t \mid X_i, \theta_g^{(z)}) = \lambda_{gj}^{(z)} \exp(X_i^\top \gamma_g^{(z)}) \quad \text{for } t\in(s_{j-1},s_j],
\end{equation*}
where $\theta_g^{(z)} := (\gamma_g^{(z)},(\lambda_{gj}^{(z)})_{j \in [J]})$ for all applicable $(g,z)$. The piecewise exponential model is quite general and can accommodate various shapes of the hazard function. Time-specific covariate effects may also be introduced by letting the parameters $\gamma_g^{(z)}$ vary across the different time intervals. Moreover, when the time axis is split at each observed failure time, the piecewise exponential model is equivalent to a Cox model \citep{collett2023modelling}. On the other hand, when $J=1$, it reduces to a parametric exponential model. 
For this simulation study and for the application in Section~\ref{sec:nivas}, we partition the time axis in three intervals: $(0,7]$, $(7,14]$, $(14,\infty)$. 

We use standard Normal priors for the logarithm of the baseline hazards $\log(\lambda_{gj}^{(z)})$ and for all the other parameters involved in the models, assuming that all the parameters are a priori independent.

\subsection{Data generation and scenarios}

We evaluate the methodologies across four scenarios. Scenarios 1 and 2 consider a baseline setting where Assumptions~\ref{ass:sutva}--\ref{ass:cond-mono}, \ref{ass:vw}, and \ref{ass:cox} hold, exploring both a large sample size ($n=2000$, Scenario 1) and a small sample size ($n=200$, Scenario 2). We then benchmark the robustness of the two methodologies against violations of either Assumption~\ref{ass:vw} (Scenario 3) or Assumption~\ref{ass:cond-mono} (Scenario 4), both under a large sample size ($n=2000$). 

We independently generate three covariates $X_i := (U_i,V_i,W_i)$ as $U_i \sim \mathrm{N}(0,1)$ and $V_i,W_i \sim \mathrm{Bernoulli}(0.5)$. Following \citet{wang2024causal}, we then assign the treatment $Z_i$ with probability $\bbP(Z_i=1 \mid X_i) = \mathrm{expit}(-1+U_i/2+V_i/2-W_i/2)$. 
The uncured probabilities are generated as $p^{(1)}(U_i,V_i) = \mathrm{expit}(U_i/4-V_i/2)$ and $p^{(0)}(U_i,V_i) = \mathrm{expit}(-1/2-U_i/3+V_i)$.
From these, we obtain the principal strata probabilities $\pi_{g}(U_i,V_i)$ as described in Section~\ref{sec:models-priors-simulation}, where we set $\rho=1$ in scenarios where Assumption~\ref{ass:cond-mono} holds, and $\rho=0.5$ to assess the behavior under violation. The principal stratum is then drawn as $G_i \sim \mathrm{Multinomial}(\pi_{g}(U_i,V_i))$ for $g \in \mathcal{G}$. 
Assumption~\ref{ass:vw} is satisfied when the covariate $V_i$ influences the potential uncured status $B_i(z)$ but not the potential time-to-event outcome $T_i(z)$. We therefore simulate $T_i(z)$, as long as it exists, from an exponential distribution with rate $\exp(\bX_i^\top \gamma_g^{(z)})/2$, where $\gamma_{\mathrm{UU}}^{(1)}=(0,-1,0,1)$, $\gamma_{\mathrm{UC}}^{(1)}=(0.5,-0.5,0,0.5)$, $\gamma_{\mathrm{UU}}^{(0)}=(1,1,0,0)$, $\gamma_{\mathrm{CU}}^{(0)}=(0.5,0.5,0,0.5)$, and where we recall that $\bX_i = (1,U_i,V_i,W_i)$. 
To generate data under violation of Assumption~\ref{ass:vw}, we consider instead the following parameters: $\gamma_{\mathrm{UU}}^{(1)}=(0,-1,0.5,1)$, $\gamma_{\mathrm{UC}}^{(1)}=(0.5,-0.5,-0.2,0.5)$, $\gamma_{\mathrm{UU}}^{(0)}=(1,1,0.2,0)$, $\gamma_{\mathrm{CU}}^{(0)}=(0.5,0.5,-0.2,0.5)$. 
Finally, the censoring times are generated as $C_i(1) \sim \text{U}(6,30)$ and $C_i(0) \sim \text{U}(5,30)$, exactly as was done in \citet{wang2024causal}. 

To replicate the method of \citet{wang2024causal}, we use the code provided by the authors. For our proposed method, we follow the model specifications described in Section~\ref{sec:models-priors-simulation}, where all the covariates $X_i=(U_i,V_i,W_i)$ are included in the models. We always set $\rho=1$ for estimation with both methods, regardless of the scenario. 

\subsection{Results}

Figure~\ref{fig:simres} presents the simulation results, comparing the point estimates of the two methods across 100 simulated datasets for the estimands $\bar{\delta}$ and $\tau_{\mathrm{UU}}^{\mathrm{RMST}}(t^*=30)$, as well as the marginal strata probabilities $\bar{\pi}_{g}$, under the four scenarios considered. Throughout, we use the posterior median as the point estimate for our Bayesian model.
In the baseline scenario with large sample size (Scenario 1), where Assumptions~\ref{ass:sutva}--\ref{ass:cond-mono}, \ref{ass:vw}, and \ref{ass:cox} all hold, both methods achieve comparable performance, confirming that our approach is competitive under ideal conditions. When the sample size is small (Scenario 2), our method generally yields lower variability at similar bias levels, reflecting the benefits of the Bayesian paradigm.
The key advantage of our framework emerges in Scenario 3, where Assumption~\ref{ass:vw} is violated. Unlike the method of \citet{wang2024causal}, which crucially relies on decompositions based on the substitutional variable for identification, our approach does not, and consequently proves more robust to violations of this assumption. This is a key practical advantage, as the substitution relevance assumption can be difficult to verify. Moreover, our method does not require the knowledge of the substitutional variable, which is challenging to select in applied settings: all covariates are included in the model without distinction. 
Finally, Scenario 4 shows that neither method is particularly robust to violations of Assumption~\ref{ass:cond-mono}. However, we discuss how different parameterizations of our model allow to relax this assumption, providing an easy way to assess sensitivity to it. 
Additional details on the simulation study can be found in Appendix~\ref{app:simulations}.

\begin{figure*}[h]
  \centering
  \includegraphics[width=\linewidth]{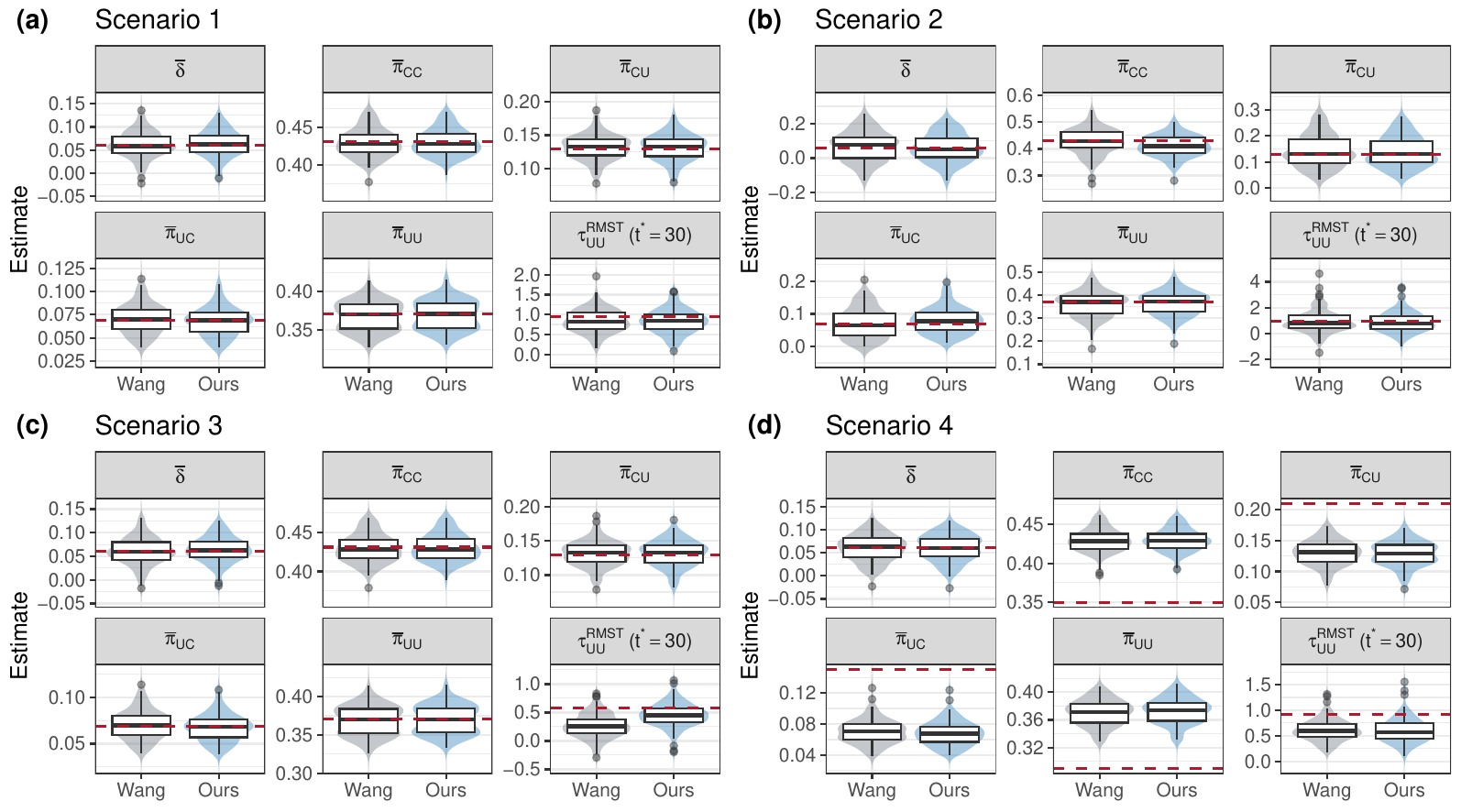}
  \caption{Simulation results for the four scenarios considered. Boxplots and violin plots represent the distribution of the point estimate of the estimands $\bar{\delta}$ and $\tau_{\mathrm{UU}}^{\mathrm{RMST}}(t^*=30)$ and the marginal strata probabilities $\bar{\pi}_{g}$ across 100 replications for the two methods. In our Bayesian model, we take the posterior median as the point estimate. The horizontal dashed line indicates the true value. Panel (a): Scenario 1 (Assumptions~\ref{ass:sutva}--\ref{ass:cond-mono}, \ref{ass:vw}, and \ref{ass:cox} hold, $n=2000$). Panel (b): Scenario 2 (Assumptions~\ref{ass:sutva}--\ref{ass:cond-mono}, \ref{ass:vw}, and \ref{ass:cox} hold, $n=200$). Panel (c): Scenario 3 (Assumption~\ref{ass:vw} violated, $n=2000$). Panel (d): Scenario 4 (Assumption~\ref{ass:cond-mono} violated, $n=2000$).}
  \label{fig:simres}
\end{figure*}

\section{Application to NIVAS} \label{sec:nivas}

The post-operative period following abdominal surgery is associated with an increased risk of partial lung collapse (atelectasis), which may lead to hypoxemia and acute respiratory failure. In patients experiencing severe respiratory deterioration, tracheal reintubation may become necessary to provide invasive mechanical ventilation and prevent further clinical decline. However, the need for reintubation is widely recognized as a marker of poor prognosis and is associated with increased morbidity, excess mortality, prolonged lengths of stay in both the intensive care unit and the hospital, and higher healthcare utilization. \

We analyze data from the NIVAS (Non-Invasive Ventilation After Surgery) study \citep{jaber2016nivas}, a randomized controlled trial aimed at assessing whether non-invasive ventilation (NIV) can prevent reintubation in patients with hypoxemic acute respiratory failure following abdominal surgery. The study includes $n=293$ patients, of which 148 are randomly assigned NIV and 145 are assigned standard oxygen therapy. 

We consider time to first reintubation or death, whichever comes first, as our outcome of interest. From a clinical perspective, reintubation typically occurs shortly after surgery among patients who experience respiratory deterioration, implying a relatively short follow-up horizon for the endpoint under study. Consequently, many patients remain event-free throughout the observation period and are not expected to require reintubation or die within this clinically relevant time window. This is also confirmed by Kaplan-Meier curves (Figure~\ref{fig:nivas_km}), showing a plateau in their right tail. Overall, these considerations justify a comparison based on cure models. 

\begin{figure}
    \centering
    \includegraphics[width=0.8\linewidth]{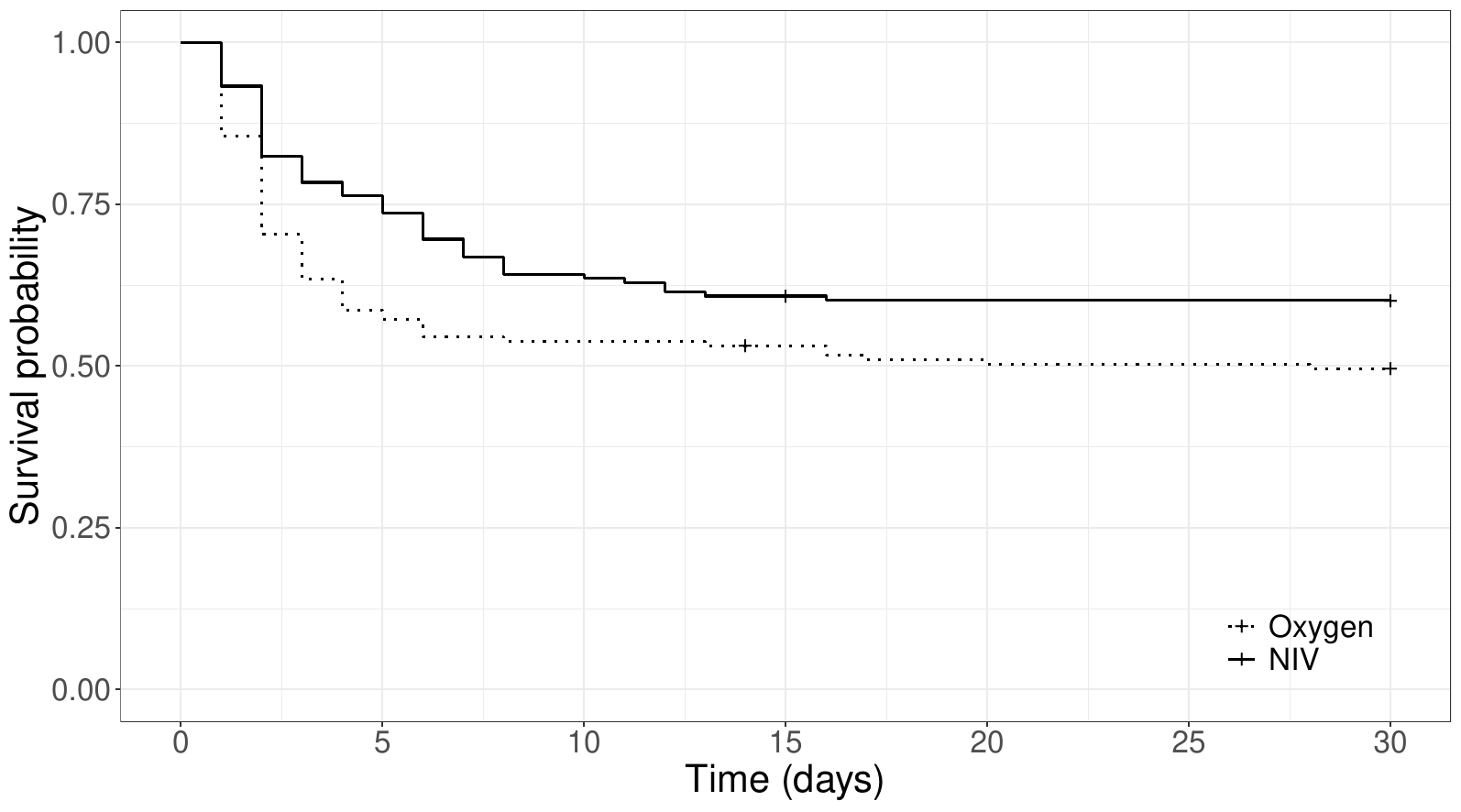}
    \caption{Kaplan-Meier curves by treatment group in the NIVAS study.}
    \label{fig:nivas_km}
\end{figure}

The goal of this analysis is threefold. First, we want to understand whether NIV can prevent patients from experiencing reintubation or death compared to standard oxygen therapy. Second, we aim at comparing survival probabilities among relevant subpopulations of patients, that is, UU and $\neg\mathrm{CC}$ patients. 
Third, we provide a characterization of principal strata based on the covariates, to inform about the observable characteristics linked to each latent principal stratum. 

\subsection{Modeling and prior details} \label{sec:nivas-models}

Within this application, we model the principal stratum variable $G_i$ through a multinomial logistic model
\begin{equation}
\label{eq:ps-model}
    \pi_g(X_i \mid \theta_G)= \frac{\exp(\bX_i^\top\beta_g)}{\sum_{k\in\cG}  \exp(\bX_i^\top \beta_k)} \quad \text{for } g\in\cG,
\end{equation}
where $\theta_G := (\beta_g)_{g\in\mathcal{G}}$. For identifiability, we set category ${\rm CC}$ as the reference, so that $\beta_{\mathrm{CC}}=0$. This is a different but more natural choice than the one presented in Section~\ref{sec:models-priors-simulation}, since it directly models the principal stratum and ensures that the category probabilities are always constrained to the probability simplex for each unit. Moreover, in general, it does not imply the monotonicity assumption. The outcome model and prior distributions remain the same as in Section~\ref{sec:models-priors-simulation}. 

We operate under the assumption that there are no patients who would be uncured if treated with NIV but cured under oxygen therapy, i.e., $\pi_{\mathrm{UC}}(X_i)=0$ for all $i\in[n]$. As already discussed in Section~\ref{sec:identification}, this is a special case of Assumption~\ref{ass:cond-mono}, which we impose by removing the category $\mathrm{UC}$ from the set $\mathcal{G}$ in Equation~\eqref{eq:ps-model}. We deem this reasonable in our study, as NIV is not generally expected to be harmful relative to standard oxygen therapy. Nonetheless, we recognize that potential NIV intolerance could challenge this assumption, and we thus perform a sensitivity analysis in Appendix~\ref{app:nivas}, presenting results obtained without this constraint. 

We also include additional information on the cure status for censored patients, therefore considering the causal cure model likelihood we derived based on Proposition~\ref{prp:likelihood-standard-cure}. In our study, censored patients are classified cured if they are discharged from the hospital. Accordingly, censoring is applied at hospital discharge, loss to follow-up, or 30 days, whichever occurs first. 
We assume the potential time to discharge within each applicable stratum $g\in\neg\mathrm{UU}$ (where $\neg\mathrm{UU} = \{\mathrm{CC,CU,UC}\}$) and treatment group $z\in\{0,1\}$ follows an exponential distribution 
\begin{equation*}
    C_i^D(z) \mid X_i, G_i=g, \theta_g^{C^D(z)} \sim \mathrm{Exp}\left( \exp\left(\boldsymbol{X}_i^\top \eta_g^{C^D(z)}\right) \right)
\end{equation*}
where $\theta_g^{C^D(z)}:=(\eta_g^{C^D(z)})_{g\in\neg\mathrm{UU},z\in\{0,1\}}$.
We show how this auxiliary information helps sharpen inference by providing the results of models that do not leverage this information in Appendix~\ref{app:nivas}. 

The analysis considers 12 baseline covariates, selected by clinical experts based on their established or suspected association with the risk of reintubation or death, and handle missing values by multiple imputation. The covariates are listed in Table~\ref{tab:ps-characterization}. Summary statistics, along with information on the imputation procedure, are provided in Appendix~\ref{app:nivas}. 

The Bayesian model is implemented in Stan. For the estimation, we run 4 chains for 2000 iterations, discarding the first 1000 as warmup. The chains mix well and there are no signs of convergence issues. 

\subsection{Results}

Posterior distributions of the relevant quantities are shown in Figure~\ref{fig:nivas-res} (a). We find that NIV has a positive but limited ability to prevent patients from reintubation or death, as represented by the posterior median of $\bar{\delta}$ of $0.03$, with $95\%$ credible interval from $0.01$ to $0.10$. 
In terms of principal stratum membership, the majority of patients in our study are either always-cured (posterior median of $\bar{\pi}_{\mathrm{CC}}$ of $0.53$, with 95\% credible interval from $0.47$ to $0.58$) or always-uncured (posterior median of $\bar{\pi}_{\mathrm{UU}}$ of $0.44$, with 95\% credible interval from $0.38$ to $0.49$). The posterior distribution $\bar{\pi}_{\mathrm{CU}}$ is structurally the same as that of $\bar{\delta}$, as a consequence of the monotonicity assumption we adopted.
These estimates suggest that approximately one patient in thirty-three may derive a preventive benefit from non-invasive ventilation with respect to reintubation or death, whereas about one-half of patients would be expected to experience a favorable outcome regardless of treatment assignment. 

The posterior median of $\tau_{\mathrm{UU}}^{\mathrm{RMST}}(t^*=30)$ indicates that NIV delays reintubation or death by $1.27$ days on average over a 30 day period among always-uncured patients. Although the associated $95\%$ credible interval ranges from $-0.76$ to $3.82$, there is a $90\%$ posterior probability that $\tau_{\mathrm{UU}}^{\mathrm{RMST}}(t^*=30)>0$, suggesting that NIV is more likely than not to be beneficial for this stratum. The advantage of NIV compared to standard oxygen therapy is even more accentuated in the union of the $\mathrm{CU}$ and $\mathrm{UU}$ strata, the $\neg\mathrm{CC}$ subgroup, where treatment with NIV results in an average gain of $2.80$ days, with a $95\%$ credible interval from $0.31$ to $6.16$, as shown by the posterior distribution of $\tau_{\neg\mathrm{CC}}^{\mathrm{RMST}}(t^*=30)$. 
The observed gain of nearly three event-free days among patients susceptible to treatment benefit is likely to be clinically meaningful, as reintubation is frequently associated with severe complications, including ventilator-associated pneumonia, prolonged intensive care unit stay, increased healthcare resource utilization, and excess mortality.
Figure~\ref{fig:nivas-res} (b) reports the posterior median and 95\% credible interval of $\tau_{\mathrm{UU}}(t)$ and $\tau_{\neg\mathrm{CC}}(t)$ over 30 days, showing the time dynamics of the difference in survival probabilities between the two treatment groups in the relevant subpopulations. In both cases, the most benefit of NIV is observed within 7 days following randomization. 

\begin{figure*}
    \centering
    \includegraphics[width=\linewidth]{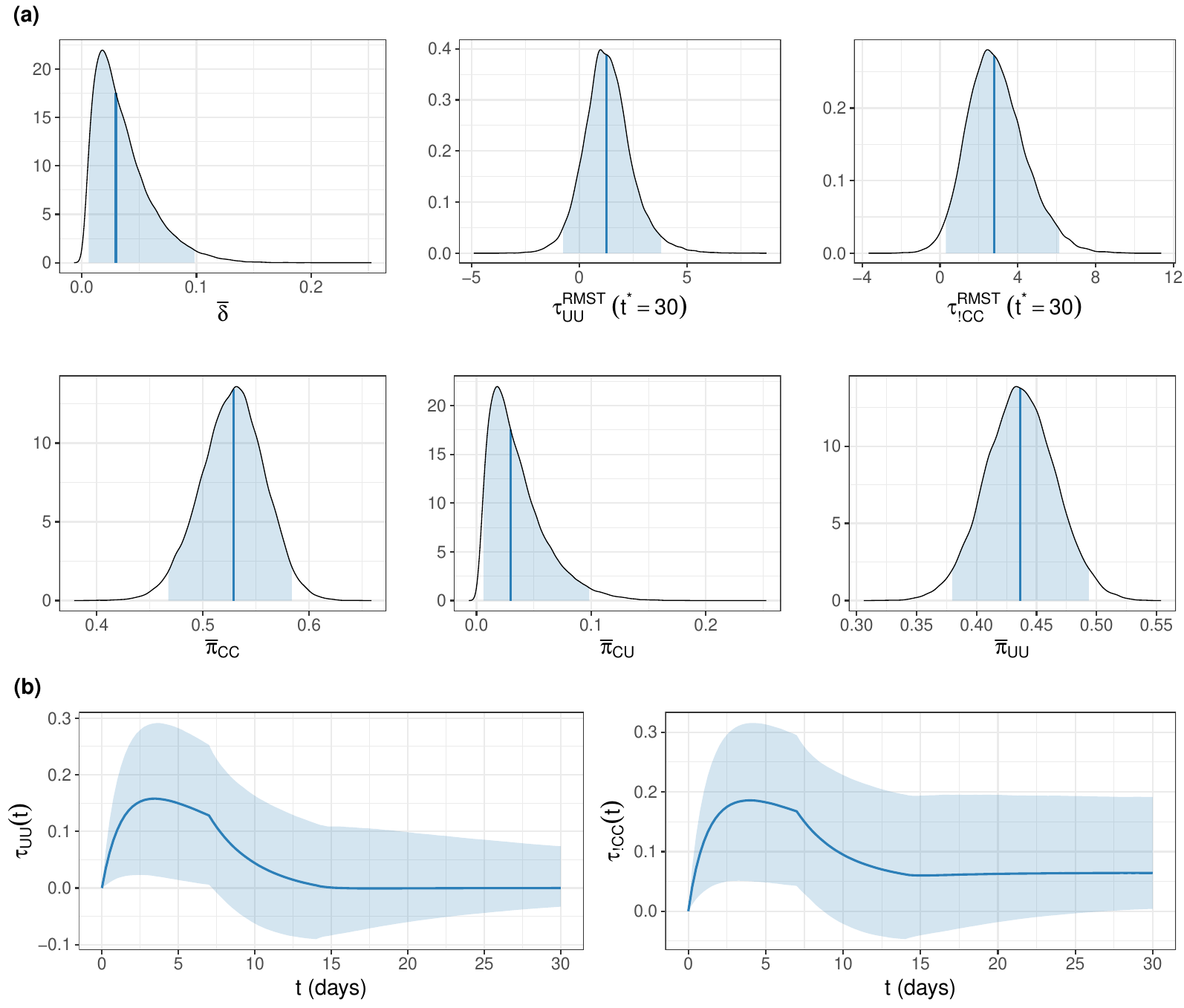}
    \caption{Posterior summaries of the relevant quantities. Panel (a): posterior distributions of the difference in cure probabilities, difference in RMST for $\mathrm{UU}$ and $\neg\mathrm{CC}$, and marginal principal strata probabilities. The vertical solid line indicates the posterior median, and the shaded area represents the 95\% credible interval. Panel (b): posterior median (solid line) and 95\% credible intervals (shaded area) of the difference in survival probabilities for $\mathrm{UU}$ and $\neg\mathrm{CC}$ patients over time.}
    \label{fig:nivas-res}
\end{figure*}

The estimated principal strata probabilities allow to evaluate the distribution of baseline covariates within principal strata. Specifically, one can obtain the mean of a given covariate $X_i$ by principal stratum as $\sum_{i=1}^n \pi_g(X_i) X_i / \sum_{i=1}^n \pi_g(X_i)$ for $g \in\mathcal{G}$. Repeating this computation across posterior draws of $\pi_g(X_i)$ yields the posterior distribution of covariates mean by principal stratum, which we summarize in Table~\ref{tab:ps-characterization}. 
We find that the $\mathrm{UU}$ stratum is characterized by a higher prevalence of patients with an IGS score $> 40$, white cell count $>$ 20000 n/$\mu$liter, who have undergone oesophagectomy, have received epidural analgesia, and experienced copious tracheal secretions. Moreover, $\mathrm{UU}$ patients generally report lower pH, lower PaO$_2$/FiO$_2$ ratio, and less extubations within 6 hours following the end of surgery compared to $\mathrm{CC}$ patients. Overall, this depicts the $\mathrm{UU}$ stratum as a more severely ill subgroup, with the baseline characteristics we just mentioned as the main drivers of discrimination between $\mathrm{CC}$ and $\mathrm{UU}$ patients. Covariate distributions for the $\mathrm{CU}$ stratum are instead associated with wide intervals, reflecting substantial uncertainty about its composition.
By linking baseline clinical characteristics to latent principal strata, the method may help identify, before treatment initiation, those patients most likely to benefit from non-invasive ventilation and thus facilitate precision-guided treatment decisions. 

\begin{table}[!h]
\centering
\caption{Posterior mean (95\% credible interval) of the mean of baseline covariates by principal stratum. For binary variables, the mean corresponds to the proportion (\%).}
\label{tab:ps-characterization}
\resizebox{\ifdim\width>\linewidth\linewidth\else\width\fi}{!}{
\begin{tabular}{lccc}
\toprule
 & Stratum $\mathrm{CC}$ & Stratum $\mathrm{CU}$ & Stratum $\mathrm{UU}$\\
\cmidrule(l{3pt}r{3pt}){2-4}
Variable & \multicolumn{3}{c}{\textit{Mean or \% (95\% credible interval)}} \\
\midrule
pH & 7.43 (7.42, 7.44) & 7.43 (7.38, 7.46) & 7.40 (7.40, 7.41)\\
PaO$_2$/FiO$_2$ (mm Hg) & 203.20 (195.21, 211.35) & 180.49 (132.56, 232.06) & 183.75 (173.66, 194)\\
Male gender & 75\% (71\%, 79\%) & 49\% (16\%, 80\%) & 81\% (75\%, 86\%)\\
IGS score $>$ 40 & 15\% (11\%, 19\%) & 16\% (2\%, 44\%) & 33\% (28\%, 38\%)\\
Psychotropic use & 9\% (7\%, 12\%) & 11\% (1\%, 31\%) & 12\% (9\%, 16\%)\\
White cell count $>$ 20000 n/$\mu$liter & 11\% (7\%, 15\%) & 15\% (2\%, 40\%) & 20\% (15\%, 25\%)\\
Oesophagectomy & 5\% (3\%, 7\%) & 4\% (0\%, 16\%) & 12\% (9\%, 15\%)\\
Epidural analgesia & 12\% (9\%, 15\%) & 7\% (1\%, 25\%) & 19\% (15\%, 24\%)\\
Extubated$<$ 6-hr after the end of surgery & 71\% (66\%, 76\%) & 48\% (17\%, 79\%) & 57\% (51\%, 64\%)\\
Copious tracheal secretions & 35\% (30\%, 39\%) & 25\% (5\%, 55\%) & 45\% (39\%, 51\%)\\
Age $\geq 60$ & 60\% (56\%, 65\%) & 44\% (13\%, 77\%) & 65\% (59\%, 70\%)\\
Upper abdominal surgery & 61\% (57\%, 66\%) & 49\% (18\%, 79\%) & 66\% (60\%, 71\%)\\
\bottomrule
\end{tabular}
}
\end{table}

The conclusions are quite robust to the monotonicity assumption we imposed. Moreover, inclusion of additional information on the cure status proves helpful to obtain a more precise inference. See Appendix~\ref{app:nivas} for a discussion on these points. 

From a precision medicine perspective, these results suggest that principal stratification may serve not only as a causal inference tool but also as a patient-clustering framework, enabling the identification of clinically meaningful subgroups with heterogeneous responses to NIV. Such information could ultimately help clinicians target NIV to patients who are most likely to derive benefit while avoiding unnecessary interventions in lower-risk individuals.

\section{Discussion}

In this paper, we proposed a Bayesian model-based approach to causal inference for cure models, leveraging the principal stratification framework to define meaningful estimands. Our method proves competitive compared to that of \citet{wang2024causal} in numerical experiments, and more applicable in practice because it does not require the choice of a substitutional variable. Sensitivity analyses to key assumptions can be easily performed, as illustrated with an application to the NIVAS study. Within this context, we show that inclusion of additional information on the cure status for censored patients can sharpen inference. Moreover, we highlight how the characterization of principal strata with observed baseline covariates is a useful tool to inform clinicians on the composition of these latent subgroups.
Since our approach relies on (semi-)parametric models, future extensions include implementation of non-parametric methods.

\paragraph{Acknowledgments}
We are grateful to Catherine Legrand for her time  and insightful discussions. We gratefully acknowledge the investigators of the NIVAS study for granting access to the data used in this work.

\bibliographystyle{apalike} 
\bibliography{references}

\clearpage
\appendix


\part*{Appendix}

\section{Proofs}
\label{app:proofs}

\subsection{Proofs of Section~\ref{sec:identification}}

\begin{proof}[Proof of Proposition~\ref{prp:ident_1}]
From Assumptions~\ref{ass:sutva}--\ref{ass:positivity}, we get 
\begin{align*}
S^{(z)}(t\mid X_i) &= \bbP(T_i(z) > t \mid X_i)  = \bbP(T_i(z) > t \mid X_i, Z_i=z) \\ 
&= \bbP(T_i > t \mid X_i, Z_i=z).
\end{align*}
The identifiability of $\bbP(T_i > t \mid X_i, Z_i=z)$ follows from standard survival analysis arguments using Assumptions~\ref{ass:cond-indep-cens} and \ref{ass:suff-large-cens}. $S^{(z)}(t)$ is identifiable from the identity $S^{(z)}(t) = \bbE[S^{(z)}(t\mid X_i)]$.

Under the additional assumption that $t^* \geq t_{\max}^{(z)}(X_i)$, it holds that $T_i(z) > t^*$ is equivalent to $B_i(z)=0$ given $X_i$, so that
$$
S^{(z)}(t^*\mid X_i) = \bbP(B_i(z)=0\mid X_i) = 1-p^{(z)}(X_i),
$$
and thus both
$
\delta(X_i)=(1-p^{(1)}(X_i))-(1-p^{(0)}(X_i))
$
and $\bar\delta = \bbE[\delta(X_i)]$ are identifiable.
\end{proof} 

\begin{proof}[Proof of Proposition~\ref{prp:ident_2}]
By Proposition~\ref{prp:ident_1}, the two 
probabilities $p^{(1)}(x)$ and $p^{(0)}(x)$ are identifiable. Furthermore, by definition,
\begin{align*}
p^{(1)}(x)
=
\pi_{\rm UU}(x)+\pi_{\rm UC}(x)
\quad \text{and}\quad
p^{(0)}(x)
=
\pi_{\rm UU}(x)+\pi_{\rm CU}(x).
\end{align*}
By Assumption~\ref{ass:cond-mono}, it holds $
\pi_{\rm UU}(x)=\min(p^{(1)}(x),p^{(0)}(x))
$. Plugging this expression into the first identities, we get
\begin{align*}
\pi_{\rm UC}(x) &=
p^{(1)}(x)-\pi_{\rm UU}(x) = \max(0,p^{(1)}(x)-p^{(0)}(x)), \\
\pi_{\rm CU}(x)
&= p^{(0)}(x)-\pi_{\rm UU}(x) = \max(0,p^{(0)}(x)-p^{(1)}(x)), \\
\pi_{\rm CC}(x)
&= 1-\pi_{\rm UU}(x)-\pi_{\rm UC}(x)-\pi_{\rm CU}(x) =
1-\max(p^{(1)}(x),p^{(0)}(x)).
\end{align*}
Since everything quantity in the RHSs is identifiable, so are the principal strata proportions $\pi_{g}(x)$ for $g \in \cG$.
\end{proof}

\begin{proof}[Proof of Proposition~\ref{prp:ident_3}] The proof of the identification under Assumption~\ref{ass:vw} already appears in \citet{wang2024causal} but we provide it here for completeness and also because it allows introducing the right notations for the proof of Proposition~\ref{prp:ident_4}. Thanks to Propositions~\ref{prp:ident_1} and \ref{prp:ident_2}, $S^{(z)}(t\mid X_i=x)$ and $\pi_g(x)$ are identifiable for $z \in \{0,1\}$ and $g \in \cG$. Furthermore, it holds
\begin{equation}
\label{eq:identproof}
S^{(z)}(t\mid X_i=x)= \pi_{\rm UU}(x) S_{\rm UU}^{(z)}(t \mid X_i=x) + \pi_{g_z}(x) S_{g_z}^{(z)}(t \mid X_i=x) + 1- p^{(z)}(x),
\end{equation}
where $g_1 = {\rm UC}$ and $g_0 = {\rm CU}$. By Assumption~\ref{ass:vw}, $S_{g}^{(z)}(t \mid X_i=x) = S_{g}^{(z)}(t \mid U_i=u)$ for $g \in \{{\rm UU}, g_z\}$. Furthermore, $\pi_{\rm UU}(x)+\pi_{g_z}(x) = p^{(z)}(x)$ and 
$$
q(x) := \frac{\pi_{\rm UU}(x)}{p^{(z)}(x)} = \bbP(B_i(1-z)=1 \mid X_i=x, B_i(z)=1),
$$
so that, also by Assumption~\ref{ass:vw}, for a given value of $U_i=u$, we can always find two values $v_1$ and $v_2$ such that $q(u,v_1)\neq q(u,v_2)$. Letting 
$$
r(x) := \frac{S^{(z)}(t \mid X_i=x)-1+p^{(z)}(x)}{p^{(z)}(x)},
$$
we thus arrive at the following set of equation for a given $u$:
\begin{equation*}
    \begin{cases}
    r(u,v_1) =& q(u,v_1) S_{\rm UU}^{(z)}(t \mid U_i=u)+(1-q(u,v_1)) S_{g_z}^{(z)}(t \mid U_i=u), \\ 
    r(u,v_2) =& q(u,v_2) S_{\rm UU}^{(z)}(t \mid U_i=u)+(1-q(u,v_2)) S_{g_z}^{(z)}(t \mid U_i=u).
\end{cases}
\end{equation*}
Since $q(u,v_1) \neq q(u,v_2)$, the above linear system is invertible, and both $S_{\rm UU}^{(z)}(t \mid U_i=u)$ and $S_{g_z}^{(z)}(t \mid U_i=u)$ are identifiable, which ends the proof. 
\end{proof}

{
\renewcommand{\theproposition}{S4}
\begin{proposition}[Formal version of Proposition~\ref{prp:ident_4}] \label{prp:ident_4_formal}
    We assume that $t^* \geq t_{\max}^{(z)}(X_i)$ a.s. for $z \in \{0,1\}$ and that the strata probabilities $\pi_g(\cdot)$ for $g \in \cG$ are in $\cC(\cX)$, the space of continuous function of $\cX$, endowed with the compact-open topology. 
    
    Then, under Assumptions~\ref{ass:sutva}--\ref{ass:cond-mono}, and \ref{ass:cox}, there exists a meagre subset $\cE$ of $\cC(\cX)$ (that is, a countable union of subset whose closure has an empty interior) such that: for $g \in \{{\rm UC,CU}\}$, if either $\supp \pi_{\rm UU} \neq \supp \pi_{g}$ or $\pi_{g}/\pi_{\rm UU} \notin \cE$, then the stratum-specific (conditional) survival functions are all identifiable. 

    Furthermore, $\cE$ has the property that if the strata probabilities $\pi_g$ follow a multinomial-logistic model with intercept $a_g$ and slope $b_g$, then $\pi_g/\pi_{\rm UU}$ is not in  $\cE$ as soon as $b_g \neq b_{\rm UU}$. 
\end{proposition}
}

\begin{proof}[Proof of Proposition~\ref{prp:ident_4_formal}]
Under Assumption~\ref{ass:cox}, it holds
$$
S^{(z)}_g(t\mid X_i=x) = \exp\left(-\Lambda^{(z)}_g(t) \exp\left(x^\top \gamma^{(z)}_g\right)\right) 
$$
for $g \in \{{\rm UU}, g_z\}$ and for some vector $\gamma^{(z)}_g \in \bbR^d$ and nonnegative and non identically zero function $\Lambda_g^{(z)}$. Assume that there exist two different sets 
$$\{\Lambda_g^{(z)},\gamma_g^{(z)}\mid g \in \{{\rm UU}, g_z\}, z\in \{0,1\}\} \quad \text{and}\quad \{\wt\Lambda_g^{(z)},\wt\gamma_g^{(z)}\mid g \in \{{\rm UU}, g_z\}, z\in \{0,1\}\},
$$ 
leading to the same observable data distribution. By Equation~\eqref{eq:identproof}, it holds
$$
\pi_{\rm UU}(x)(S^{(z)}_{\rm UU}(t\mid X_i=x)-\wt S^{(z)}_{\rm UU}(t\mid X_i=x)) = \pi_{g_z}(x)(\wt S^{(z)}_{g_z}(t\mid X_i=x)-S^{(z)}_{g_z}(t\mid X_i=x)). 
$$
Let us first assume that $\supp \pi_{\rm UU} \neq \supp \pi_{g_z}$. We can assume without loss of generality that there exists an open subset on which $\pi_{\rm UU}(x)>0$ and $\pi_{g_z}(x)=0$. On this subset, it holds $S^{(z)}_{\rm UU}(t\mid X_i=x)=\wt S^{(z)}_{\rm UU}(t\mid X_i=x)$ for all $t$, from which we deduce that $\gamma_{\rm UU}^{(z)} = \wt \gamma_{\rm UU}^{(z)}$ and $\Lambda_{\rm UU}^{(z)} = \wt \Lambda_{\rm UU}^{(z)}$. Therefore, $S^{(z)}_{\rm UU}(t\mid X_i=x) = \wt S^{(z)}_{\rm UU}(t\mid X_i=x)$ on the whole $\cX$ which implies that 
$S^{(z)}_{g_z}(t\mid X_i=x) = \wt S^{(z)}_{g_z}(t\mid X_i=x)$ on the support of $\pi_{g_z}$, leading to $\gamma_{g_z}^{(z)} = \wt \gamma_{g_z}^{(z)}$ and $\Lambda_{g_z}^{(z)} = \wt \Lambda_{g_z}^{(z)}$, and a contradiction. 

Assume now that $\supp \pi_{\rm UU} = \supp \pi_{g_z}$. For $\omega = (\gamma_1,\tilde\gamma_1, \gamma_2,\tilde\gamma_2,\lambda_1,\tilde\lambda_1,\lambda_2,\tilde\lambda_2) \in \bbR^{4d+4}$, we let 
$$
Q_\omega(x) := \frac{\exp(-\lambda_1 e^{\gamma_1^\top x}) - \exp(-\tilde\lambda_1 e^{\tilde\gamma_1^\top x})}{\exp(-\tilde\lambda_2 e^{\tilde\gamma_2^\top x})-\exp(-\lambda_2 e^{\gamma_2^\top x})} =: \frac{N_\omega(x)}{D_\omega(x)},
$$
which is locally defined around all $x$ that are such that $D_\omega(x) \neq 0$. For any compact $K \subset \cX$ with non-empty interior, let 
$$
\Omega_K := \{\omega \in \bbR^{4d+4} ~\middle|~ \inf_{x \in K} |D_\omega(x)| \neq 0\}. 
$$
There exists an exhaustion of compact set $\cL_K$ such that $\Omega_K$ is the increasing union of all sets $L$ in $\cL_K$. Thus, letting $\Phi_K: \omega \in \Omega_K \mapsto \restriction{Q_\omega}{K} \in \cC(K)$, it holds
$$
\Phi_K(\Omega_K) = \bigcup_{L \in \cL_K} \Phi_K(L).
$$
Since $\Phi_K$ is continuous, $\Phi_K(L)$ is compact, hence meagre in the infinite-dimensional space $\cC(K)$. It results that $\Phi_K(\Omega_K)$ is meagre too.

Now let $t > 0$ for which $\Lambda_{g_z}^{(z)}(t)$ and $\Lambda_{\rm UU}^{(z)}(t)$ are not zero. Let $\rho(x) := \pi_{g_z}(x)/\pi_{\rm UU}(x)$ which is not zero and well defined on $x \in \cX$. If $\wt S^{(z)}_{g_z}(t\mid X_i=x)=S^{(z)}_{g_z}(t\mid X_i=x)$ for all $x\in\cX$, then also $\wt S^{(z)}_{\rm UU}(t\mid X_i=x)=S^{(z)}_{\rm UU}(t\mid X_i=x)$ which imply $\gamma_{g_z}^{(z)} = \wt \gamma_{g_z}^{(z)}$, $\gamma_{\rm UU}^{(z)} = \wt \gamma_{\rm UU}^{(z)}$, $\Lambda_{g_z}(t)=\wt \Lambda_{g_z}(t)$ and $\Lambda_{\rm UU}(t) = \wt \Lambda_{\rm UU}(t)$. By assumption, there must thus be a such $t$ such that $\wt S^{(z)}_{g_z}(t\mid X_i=x) \neq S^{(z)}_{g_z}(t\mid X_i=x)$. We let $\cK$ be countable basis of compact sets for the usual topology on $\cX$. By the reasoning above, there exists a compact neighborhood $K \in \cK$ of $x$ such that $\restriction{\rho}{K} \in \Phi_K(\Omega_K)$. Hence 
$$
\rho \in \cE := \bigcup_{K \in \cK} r_K^{-1}(\Phi_{K}(\Omega_K)),
$$
where $r_K$ is the restriction to $K$. Since $r_K : \cC(\cX) \to \cC(K)$ is a linear, continuous and surjective map between Frechet spaces, it is open by the open mapping theorem, thus $r_K^{-1}(\Phi_{K}(\Omega_K))$ is meagre, hence $\cE$ too.

Finally, assume the $\pi_g(\cdot)$ follows a multinomial logistic model with intercept $a_g$ and slope $b_g$. Then $\pi_g(x)/\pi_{\rm UU}(x)  = c e^{\eta^\top x}$ with $c = e^{a_g-a_{\rm UU}}$ and $\eta = b_g-b_{\rm UU}$. By Equation~\eqref{eq:identproof}, non-identification implies 
\begin{equation}
    e^{-\lambda_1 e^{\gamma_1^\top x}} +c e^{\eta^\top x} e^{-\lambda_2 e^{\gamma_2^\top x}} - e^{-\wt \lambda_1 e^{\wt \gamma_1^\top x}} - c e^{\eta^\top x} e^{-\wt\lambda_2 e^{\wt\gamma_2^\top x}} = 0, \quad \forall x\in\cX, \label{eq:mlog}
\end{equation}
for the suitable vector of parameters $(\gamma_1,\tilde\gamma_1, \gamma_2,\tilde\gamma_2,\lambda_1,\tilde\lambda_1,\lambda_2,\tilde\lambda_2)$. We show that Equation~\eqref{eq:mlog} with $\eta \neq 0$ and $\lambda_1,\lambda_2 > 0$ implies that $\lambda_i = \wt\lambda_i$ and $\gamma_i = \wt\gamma_i$ for $ i \in \{1,2\}$, which will be a contradiction. Since Equation~\eqref{eq:mlog} holds in an open subset of $\bbR^d$, and because every function involved is entire, it also holds on $\bbC^d$. Now Equation~\eqref{eq:mlog} is a null linear combination of exponential of function of the form:
$$
f : \bz \in \bbC^d \mapsto A \eta^\top \bz + B e^{\gamma^\top \bz}.
$$
By a result of \citet{narasimhan1971analogue}, this can only be the case if at least two of the four exponentiated functions of Equation~\eqref{eq:mlog} differ by a complex number, that is if 
$$
(A-A')\eta^\top \bz + B e^{\gamma^\top \bz} - B' e^{\gamma'^\top \bz} = z_0 \quad \text{for some $z_0 \in \bbC$ and all $\bz\in\bbC^d$},
$$
and for the corresponding parameters $A,B,\gamma,A',B',\gamma'$.
That can only be the case if $A=A'$, $\gamma=\gamma'$ and $B=B'$. This allows to prove that $\gamma_1 = \wt \gamma_1$, $\gamma_2 = \wt \gamma_2$, $\lambda_1 = \wt \lambda_1$ and $\lambda_2 = \wt \lambda_2$, ending the proof.
\end{proof}


\subsection{Proofs of Section~\ref{sec:bayes}}

\begin{proof}[Proof of Theorem~\ref{thm:likelihood-standard}]
For each subject $i$ we observe the random variables $(Z_i,X_i,Y_i,\Delta_i)$. We assume the joint distribution of these variables of all units is governed by a generic parameter $\theta$, conditional on which the random variables for each unit are i.i.d. Let $O_n := (Z_i,X_i,Y_i,\Delta_i)_{i=1}^n$. With a slight abuse of notation, we use $\mathbb{P}(\cdot \mid \cdot)$ to denote either a probability for discrete variables or a probability density for continuous variables. The joint distribution of the variables factorizes as:
{\allowdisplaybreaks
\begin{align*}
    \mathbb{P}(O_n) =& \prod_{i=1}^n \mathbb{P}(Z_i,X_i,Y_i,\Delta_i) \\
    =& \prod_{i=1}^n \mathbb{P}(Y_i,\Delta_i \mid X_i,Z_i) \mathbb{P}(Z_i \mid X_i) \mathbb{P}(X_i) \\
    \propto& \prod_{i=1}^n \mathbb{P}(Y_i,\Delta_i \mid X_i,Z_i) \\
    \propto& \prod_{i=1}^n \mathbb{P}(T_i=Y_i \mid X_i,Z_i)^{\Delta_i} \mathbb{P}(T_i>Y_i \mid X_i,Z_i)^{1-\Delta_i} \\
    =& \prod_{i=1}^n \left[ \mathbb{P}(T_i=Y_i \mid X_i,Z_i,B_i=1) \mathbb{P}(B_i=1 \mid X_i,Z_i) \right]^{\Delta_i} \\
    &\quad \times \left[ \mathbb{P}(T_i>Y_i \mid X_i,Z_i,B_i=1) \mathbb{P}(B_i=1 \mid X_i,Z_i) + \mathbb{P}(B_i=0 \mid X_i, Z_i) \right]^{1-\Delta_i} \\
    =& \prod_{z\in\{0,1\}} \prod_{i:Z_i=z} \left[ \mathbb{P}(T_i=Y_i \mid X_i,Z_i=z,B_i=1) \mathbb{P}(B_i=1 \mid X_i,Z_i=z) \right]^{\Delta_i} \\
    &\quad \times \left[ \mathbb{P}(T_i>Y_i \mid X_i,Z_i=z,B_i=1) \mathbb{P}(B_i=1 \mid X_i,Z_i=z) + \mathbb{P}(B_i=0 \mid X_i, Z_i=z) \right]^{1-\Delta_i} 
\end{align*}
}
The first proportional sign holds because $\mathbb{P}(Z_i \mid X_i)$ and $\mathbb{P}(X_i)$ do not depend on the parameters governing the outcome model, which are the objective of our inference.
The second proportional sign holds because, combining Assumptions~\ref{ass:sutva} and \ref{ass:cond-indep-cens} (which guarantee that $C_i$ is conditionally independent of $T_i$ given $X_i, Z_i$), we have that 
\begin{equation*}
    \mathbb{P}(Y_i,\Delta_i=1 \mid X_i,Z_i) 
    = \mathbb{P}(T_i=Y_i, C_i \geq Y_i \mid X_i,Z_i) 
    \propto \mathbb{P}(T_i=Y_i \mid X_i,Z_i),
\end{equation*}
\begin{equation*}
    \mathbb{P}(Y_i,\Delta_i=0 \mid X_i,Z_i)
    = \mathbb{P}(T_i>Y_i, C_i=Y_i \mid X_i,Z_i) 
    \propto \mathbb{P}(T_i>Y_i \mid X_i,Z_i).
\end{equation*}
The third equality holds because, recalling that $T_i=\infty$ for $B_i=0$, and $Y_i<\infty$: 
\begin{equation*}
    \mathbb{P}(T_i=Y_i \mid X_i,Z_i) 
    = \mathbb{P}(T_i=Y_i \mid X_i,Z_i,B_i=1) \mathbb{P}(B_i=1 \mid X_i,Z_i) + 0 \times \mathbb{P}(B_i=0 \mid X_i,Z_i),
\end{equation*}
\begin{equation*}
    \mathbb{P}(T_i>Y_i \mid X_i,Z_i) 
    = \mathbb{P}(T_i>Y_i \mid X_i,Z_i,B_i=1) \mathbb{P}(B_i=1 \mid X_i,Z_i) + 1 \times \mathbb{P}(B_i=0 \mid X_i,Z_i). 
\end{equation*}
Under Assumptions~\ref{ass:sutva} and \ref{ass:ignorability}, for $z\in\{0,1\}$, we have that: 
(i) $\mathbb{P}(T_i=Y_i \mid X_i,Z_i=z,B_i=1) = \mathbb{P}(T_i(z)=Y_i \mid X_i,B_i(z)=1) = f_{\mathrm{U}}^{(z)}(Y_i \mid X_i)$; 
(ii) $\mathbb{P}(T_i>Y_i \mid X_i,Z_i=z,B_i=1) = \mathbb{P}(T_i(z)>Y_i \mid X_i,B_i(z)=1) = S_{\mathrm{U}}^{(z)}(Y_i \mid X_i)$; 
(iii) $\mathbb{P}(B_i=1 \mid X_i, Z_i=z) = \mathbb{P}(B_i(z)=1 \mid X_i) = p^{(z)}(X_i)$. By substituting these quantities to the final form of $\mathbb{P}(O_n)$ we obtain the likelihood in Theorem~\ref{thm:likelihood-standard}.
\end{proof}

\begin{proof}[Proof of Equations~\eqref{eq:peq1}--\eqref{eq:Seq2}]
Under treatment, the probability of being uncured can be written as 
{\allowdisplaybreaks
\begin{align*}
    p^{(1)}(X_i) &= \mathbb{P}(B_i(1)=1 \mid X_i) \\
    &= \mathbb{P}(B_i(1)=1 \mid X_i, B_i(0)=1) \mathbb{P}(B_i(0)=1 \mid X_i) \\
    &\quad + \mathbb{P}(B_i(1)=1 \mid X_i, B_i(0)=0) \mathbb{P}(B_i(0)=0 \mid X_i) \\
    &= \mathbb{P}(B_i(1)=1, B_i(0)=1 \mid X_i) + \mathbb{P}(B_i(1)=1, B_i(0)=0 \mid X_i) \\
    &= \pi_{\mathrm{UU}}(X_i) + \pi_{\mathrm{UC}}(X_i) ,
\end{align*}
}
which is Equation~\eqref{eq:peq1}. Equation~\eqref{eq:peq2} can be derived with the same rationale. Similarly, the survival probability for the uncured under treatment is 
{\allowdisplaybreaks
\begin{align*}
    S^{(1)}_{\mathrm{U}}(t \mid X_i) &= \mathbb{P}(T_i(1)>t \mid X_i, B_i(1)=1) \\
    &= \mathbb{P}(T_i(1)>t \mid X_i, B_i(1)=1, B_i(0)=1) \mathbb{P}(B_i(0)=1 \mid X_i, B_i(1)=1) \\
    &\quad + \mathbb{P}(T_i(1)>t \mid X_i, B_i(1)=1, B_i(0)=0) \mathbb{P}(B_i(0)=0 \mid X_i, B_i(1)=1) \\
    &= S_{\mathrm{UU}}^{(1)}(t \mid X_i) \frac{\mathbb{P}(B_i(1)=1, B_i(0)=1 \mid X_i)}{\mathbb{P}(B_i(1)=1 \mid X_i)} \\
    &\quad + S_{\mathrm{UC}}^{(1)}(t \mid X_i) \frac{\mathbb{P}(B_i(1)=1, B_i(0)=0 \mid X_i)}{\mathbb{P}(B_i(1)=1 \mid X_i)} \\
    &= \frac{\pi_{\mathrm{UU}}(X_i) S_{\mathrm{UU}}^{(1)}(t\mid X_i) + \pi_{\mathrm{UC}}(X_i) S_{\mathrm{UC}}^{(1)}(t\mid X_i)}{\pi_{\mathrm{UU}}(X_i) + \pi_{\mathrm{UC}}(X_i)} ,
\end{align*}
}
proving Equation~\eqref{eq:Seq1}. Equation~\eqref{eq:Seq2} follows the same reasoning.
\end{proof}

\begin{proof}[Proof of Equations~\eqref{eq:Sg} and \eqref{eq:Sunion}]
Let $\mathcal{X}$ denote the support of $X_i$, and $f(x)$ its probability density function. First, for a given stratum $g\in\cG$, we have:
{\allowdisplaybreaks
\begin{align*}
    S_{g}^{(z)}(t) 
    &= \mathbb{P}(T_i(z)>t \mid G_i=g) \\
    &= \mathbb{E}\left[ \mathbb{I}(T_i(z)>t) \mid G_i=g \right] \\
    &= \mathbb{E}\left[ \mathbb{E}\left[ \mathbb{I}(T_i(z)>t) \mid X_i,G_i=g \right] \mid G_i=g \right] \\
    &= \mathbb{E}\left[ \mathbb{P}(T_i(z)>t \mid X_i,G_i=g) \mid G_i=g \right] \\
    &= \int_{\mathcal{X}} \mathbb{P}(T_i(z)>t \mid X_i=x,G_i=g) f(x \mid G_i=g) \, \mathrm{d}x \\
    &= \int_{\mathcal{X}} \mathbb{P}(T_i(z)>t \mid X_i=x,G_i=g) \frac{\mathbb{P}(G_i=g \mid X_i=x)}{\mathbb{P}(G_i=g)} f(x) \, \mathrm{d}x \\
    &= \mathbb{E}\left[ \frac{\pi_g(X_i)}{\bar\pi_g} S_g^{(z)}(t \mid X_i) \right] ,
\end{align*}
}
where the third equality follows from the law of iterated expectations, and the sixth equality follows from Bayes theorem. This proves Equation~\eqref{eq:Sg}.
Second, for any subset $K\subset\mathcal{G}$:
{\allowdisplaybreaks
\begin{align*}
    S_{K}^{(z)}(t) 
    &= \mathbb{P}(T_i(z)>t \mid G_i\in K) \\
    & = \mathbb{E}\left[ \mathbb{I}(T_i(z)>t) \mid G_i\in K \right] \\
    & = \mathbb{E}\left[ \mathbb{E}\left[ \mathbb{I}(T_i(z)>t) \mid X_i,G_i\in K \right] \mid G_i\in K \right] \\
    & = \mathbb{E}\left[ \mathbb{P}(T_i(z)>t \mid X_i,G_i\in K) \mid G_i\in K \right] \\
    & = \int_{\mathcal{X}} \mathbb{P}(T_i(z)>t \mid X_i=x,G_i\in K) f(x \mid G_i\in K) \, \mathrm{d}x \\
    & = \int_{\mathcal{X}} \mathbb{P}(T_i(z)>t \mid X_i=x,G_i\in K) \frac{\mathbb{P}(G_i\in K \mid X_i=x)}{\mathbb{P}(G_i\in K)} f(x) \, \mathrm{d}x \\
    & = \mathbb{E}\left[ \frac{\mathbb{P}(G_i\in K \mid X_i)}{\mathbb{P}(G_i\in K)} \mathbb{P}(T_i(z)>t \mid X_i,G_i\in K) \right] \\
    & = \mathbb{E}\left[ \frac{\sum_{g\in K} \mathbb{P}(T_i(z)>t \mid G_i=g,X_i) \mathbb{P}(G_i=g \mid X_i)}{\sum_{g\in K} \mathbb{P}(G_i=g)} \right] \\
    & = \mathbb{E}\left[ \frac{\sum_{g\in K} \pi_g(X_i) S_g^{(z)}(t \mid X_i)}{\sum_{g\in K} \bar\pi_g} \right] , 
\end{align*}
}
where the third equality follows from the law of iterated expectations, the sixth equality from Bayes theorem, and the eighth equality applies the law of total probability on the restricted support $K$.
This proves Equation~\eqref{eq:Sunion}.
\end{proof}

\begin{proof}[Proof of Proposition~\ref{prp:likelihood-standard-cure}]
We recall that in this proof any reference to Assumptions~\ref{ass:sutva}, \ref{ass:ignorability}, \ref{ass:cond-indep-cens} is intended as a reference to their expanded versions that include $C^D_i$, as clarified in the main text.
For each subject $i$ we observe the random variables $(Y_i,\Delta_i,Z_i,X_i,D_i)$. We assume the joint distribution of these variables of all units is governed by a generic parameter $\theta$, conditional on which the random variables for each unit are i.i.d.. Let $O_n^+ = (Y_i,\Delta_i,Z_i,X_i,D_i)_{i=1}^n$. With a slight abuse of notation, we use $\mathbb{P}(\cdot \mid \cdot)$ to denote either a probability for discrete variables or a probability density for continuous variables.
The joint distribution of the variables factorizes as:
{\allowdisplaybreaks
\begin{align*}
    \mathbb{P}(O_n^+) =& \prod_{i=1}^n \mathbb{P}(Y_i,\Delta_i,D_i,Z_i,X_i) \\
    =& \prod_{i=1}^n \mathbb{P}(Y_i,\Delta_i,D_i \mid X_i,Z_i) \mathbb{P}(Z_i \mid X_i) \mathbb{P}(X_i) \\
    \propto& \prod_{i=1}^n \mathbb{P}(Y_i,\Delta_i,D_i \mid X_i,Z_i) \\
    \propto& \prod_{i=1}^n [\mathbb{P}(T_i=Y_i \mid X_i,Z_i,B_i=1) \mathbb{P}(B_i=1 \mid X_i,Z_i)] ^{\Delta_i} \\
    &\quad \times [\mathbb{P}(T_i>Y_i \mid X_i,Z_i,B_i=1) \mathbb{P}(B_i=1 \mid X_i,Z_i) \\
    &\qquad + \mathbb{P}(C^D_i>Y_i \mid X_i,Z_i,B_i=0) \mathbb{P}(B_i=0 \mid X_i,Z_i)]^{(1-\Delta_i)(1-D_i)} \\ 
    &\quad \times [\mathbb{P}(C^D_i=Y_i \mid X_i,Z_i,B_i=0) \mathbb{P}(B_i=0 \mid X_i,Z_i)]^{(1-\Delta_i)D_i} \\
    =& \prod_{z\in\{0,1\}} \prod_{i:Z_i=z} [\mathbb{P}(T_i=Y_i \mid X_i,Z_i=z,B_i=1) \mathbb{P}(B_i=1 \mid X_i,Z_i=z)] ^{\Delta_i} \\
    &\quad \times [\mathbb{P}(T_i>Y_i \mid X_i,Z_i=z,B_i=1) \mathbb{P}(B_i=1 \mid X_i,Z_i=z) \\
    &\qquad + \mathbb{P}(C^D_i>Y_i \mid X_i,Z_i=z,B_i=0) \mathbb{P}(B_i=0 \mid X_i,Z_i=z)]^{(1-\Delta_i)(1-D_i)} \\ 
    &\quad \times [\mathbb{P}(C^D_i=Y_i \mid X_i,Z_i=z,B_i=0) \mathbb{P}(B_i=0 \mid X_i,Z_i=z)]^{(1-\Delta_i)D_i}
\end{align*}
}

The first proportional sign holds because $\mathbb{P}(Z_i \mid X_i)$ and $\mathbb{P}(X_i)$ do not depend on the parameters governing the outcome model, which are the objective of our inference.

The second proportional sign holds for the three following arguments (please also refer to Section~\ref{app:threecat} for the formalization of the three categories implied by $\Delta_i$ and $D_i$). 
First, when $\Delta_i=1$: 
{\allowdisplaybreaks
\begin{align*}
    \mathbb{P}(Y_i,\Delta_i=1 \mid X_i,Z_i) 
    &= \mathbb{P}(T_i=Y_i, C_i \geq Y_i \mid X_i,Z_i) \\
    &= \mathbb{P}(C_i \geq Y_i \mid X_i,Z_i) \mathbb{P}(T_i=Y_i \mid X_i,Z_i) \\
    &\propto \mathbb{P}(T_i=Y_i \mid X_i,Z_i) \\
    &=  \mathbb{P}(T_i=Y_i \mid X_i,Z_i,B_i=1) \mathbb{P}(B_i=1 \mid X_i,Z_i) ,
\end{align*}
}
where the second equality follows from combining Assumptions~\ref{ass:sutva} and \ref{ass:cond-indep-cens} (which guarantee that $C_i$ is conditionally independent of $T_i, C^D_i$ given $X_i, Z_i$), and the third equality uses the fact that $\mathbb{P}(T_i=Y_i \mid X_i,Z_i,B_i=0)=0$ (since $B_i=0 \Rightarrow T_i=\infty$, and $Y_i<\infty$). 
Second, when $(\Delta_i=0, D_i=0)$:
{\allowdisplaybreaks
\begin{align*}
    &\mathbb{P}(Y_i,\Delta_i=0, D_i=0 \mid X_i,Z_i) \\
    &\quad = \mathbb{P}(T_i>Y_i, C_i=Y_i, C^D_i>Y_i \mid X_i,Z_i) \\
    &\quad = \mathbb{P}(T_i>Y_i, C_i=Y_i, C^D_i>Y_i \mid X_i,Z_i,B_i=1) \mathbb{P}(B_i=1 \mid X_i,Z_i) \\
    &\qquad + \mathbb{P}(T_i>Y_i, C_i=Y_i, C^D_i>Y_i \mid X_i,Z_i,B_i=0) \mathbb{P}(B_i=0 \mid X_i,Z_i) \\
    &\quad = \mathbb{P}(C^D_i>Y_i \mid X_i,Z_i,B_i=1,T_i>Y_i) \mathbb{P}(C_i=Y_i \mid X_i,Z_i) \\
    &\qquad\quad \times \mathbb{P}(T_i>Y_i \mid X_i,Z_i,B_i=1) \mathbb{P}(B_i=1 \mid X_i,Z_i) \\
    &\qquad + \mathbb{P}(C^D_i>Y_i \mid X_i,Z_i,B_i=0) \mathbb{P}(C_i=Y_i \mid X_i,Z_i) \\
    &\qquad\quad \times \mathbb{P}(T_i>Y_i \mid X_i,Z_i,B_i=0) \mathbb{P}(B_i=0 \mid X_i,Z_i) \\
    &\quad =  \mathbb{P}(C_i=Y_i \mid X_i,Z_i) \mathbb{P}(T_i>Y_i \mid X_i,Z_i,B_i=1) \mathbb{P}(B_i=1 \mid X_i,Z_i) \\
    &\qquad + \mathbb{P}(C^D_i>Y_i \mid X_i,Z_i,B_i=0) \mathbb{P}(C_i=Y_i \mid X_i,Z_i) \mathbb{P}(B_i=0 \mid X_i,Z_i) \\
    &\quad \propto \mathbb{P}(T_i>Y_i \mid X_i,Z_i,B_i=1) \mathbb{P}(B_i=1 \mid X_i,Z_i) \\
    &\qquad + \mathbb{P}(C^D_i>Y_i \mid X_i,Z_i,B_i=0) \mathbb{P}(B_i=0 \mid X_i,Z_i) ,
\end{align*}
}
where the third equality follows from Assumptions~\ref{ass:sutva} and \ref{ass:cond-indep-cens} and the fact that $\{B_i=0\} \cap \{T_i>Y_i\} = \{B_i=0\}$. The fourth equality holds because $\mathbb{P}(C^D_i>Y_i \mid X_i,Z_i,B_i=1,T_i>Y_i) = 1$ (since $B_i=1 \Rightarrow C^D_i=\infty$, and $Y_i<\infty$) and $\mathbb{P}(T_i>Y_i \mid X_i,Z_i,B_i=0)=1$. 
Third, when $(\Delta_i=0, D_i=1)$: 
{\allowdisplaybreaks
\begin{align*}
    &\mathbb{P}(Y_i,\Delta_i=0, D_i=1 \mid X_i,Z_i) \\
    &\quad = \mathbb{P}(T_i>Y_i, C_i \geq Y_i, C^D_i=Y_i \mid X_i,Z_i) \\
    &\quad = \mathbb{P}(T_i>Y_i, C_i \geq Y_i, C^D_i=Y_i \mid X_i,Z_i,B_i=1) \mathbb{P}(B_i=1 \mid X_i,Z_i) \\
    &\qquad + \mathbb{P}(T_i>Y_i, C_i \geq Y_i, C^D_i=Y_i \mid X_i,Z_i,B_i=0) \mathbb{P}(B_i=0 \mid X_i,Z_i) \\
    &\quad = \mathbb{P}(C^D_i=Y_i \mid X_i,Z_i, B_i=1, T_i>Y_i) \mathbb{P}(C_i \geq Y_i \mid X_i,Z_i) \\
    &\qquad\quad \times \mathbb{P}(T_i>Y_i \mid X_i,Z_i,B_i=1) \mathbb{P}(B_i=1 \mid X_i,Z_i) \\
    &\qquad + \mathbb{P}(C^D_i=Y_i \mid X_i,Z_i, B_i=0) \mathbb{P}(C_i \geq Y_i \mid X_i,Z_i) \\
    &\qquad\quad \times \mathbb{P}(T_i>Y_i \mid X_i,Z_i,B_i=0) \mathbb{P}(B_i=0 \mid X_i,Z_i) \\
    &\quad = \mathbb{P}(C^D_i=Y_i \mid X_i,Z_i, B_i=0) \mathbb{P}(C_i \geq Y_i \mid X_i,Z_i) \mathbb{P}(B_i=0 \mid X_i,Z_i) \\
    &\quad \propto \mathbb{P}(C^D_i=Y_i \mid X_i,Z_i, B_i=0) \mathbb{P}(B_i=0 \mid X_i,Z_i) ,
\end{align*}
}
where the third equality follows from Assumptions~\ref{ass:sutva} and \ref{ass:cond-indep-cens} and the fact that $\{B_i=0\} \cap \{T_i>Y_i\} = \{B_i=0\}$, while the fourth equality uses $\mathbb{P}(C^D_i=Y_i \mid X_i,Z_i, B_i=1, T_i>Y_i)=0$ and $\mathbb{P}(T_i>Y_i \mid X_i,Z_i,B_i=0)=1$.

Under Assumptions~\ref{ass:sutva} and \ref{ass:ignorability}, for $z\in\{0,1\}$, we have that: 
(i) $\mathbb{P}(T_i=Y_i \mid X_i,Z_i=z,B_i=1) = \mathbb{P}(T_i(z)=Y_i \mid X_i,B_i(z)=1) = f_{\mathrm{U}}^{(z)}(Y_i \mid X_i)$; 
(ii) $\mathbb{P}(T_i>Y_i \mid X_i,Z_i=z,B_i=1) = \mathbb{P}(T_i(z)>Y_i \mid X_i,B_i(z)=1) = S_{\mathrm{U}}^{(z)}(Y_i \mid X_i)$; 
(iii) $\mathbb{P}(B_i=1 \mid X_i, Z_i=z) = \mathbb{P}(B_i(z)=1 \mid X_i) = p^{(z)}(X_i)$; 
(iv) $\mathbb{P}(C^D_i=Y_i \mid X_i,Z_i=z,B_i=0) = \mathbb{P}(C^D_i(z)=Y_i \mid X_i,B_i(z)=0) = f_{\neg\mathrm{U}}^{C^D(z)}(Y_i \mid X_i)$; 
(v) $\mathbb{P}(C^D_i>Y_i \mid X_i,Z_i=z,B_i=0) = \mathbb{P}(C^D_i(z)>Y_i \mid X_i,B_i(z)=0) = S_{\neg\mathrm{U}}^{C^D(z)}(Y_i \mid X_i)$. By substituting these quantities to the final form of $\mathbb{P}(O_n^+)$ we get the likelihood in Proposition~\ref{prp:likelihood-standard-cure}.
\end{proof}

\section{Classification based on $(\Delta_i,D_i)$ for auxiliary information} \label{app:threecat}

In this section, we clarify the classification scheme resulting from the combination of the event indicator $\Delta_i = \mathbb{I}(T_i \leq C_i)$ and the hospital discharge indicator $D_i = \mathbb{I}(C^D_i \leq C_i)$. Recall that $C_i \in \mathbb{R}^+$ denotes a finite censoring time ($C_i < \infty$), while $T_i, C^D_i \in \mathbb{R}^+ \cup \{\infty\}$. From the definition of $C^D_i$, we have that exactly one of $T_i$ and $C^D_i$ is finite: either $T_i < \infty$  and $C_i^D = \infty$, or $C_i^D < \infty$ and $T_i = \infty$. 
Although two binary indicators theoretically yield four combinations, one is structurally impossible under our framework, as explained below. 
\begin{enumerate}
    \item $(\Delta_i=0,D_i=0)$: patients in this category have both $T_i > C_i$ and $C^D_i > C_i$. This creates two possible situations based on which of $T_i$ and $C_i^D$ is finite (Figure~\ref{sfig:threecat}, top): either $C_i<T_i<C_i^D=\infty$, or $C_i<C_i^D<T_i=\infty$. In the first case the patient is uncured, in the second case the patient is cured. 
    \item $(\Delta_i=0,D_i=1)$: patients in this category have $T_i>C_i$ and $C^D_i \leq C_i$. Thus, it must be that $C^D_i \leq C_i < T_i=\infty$, and the patient is cured (Figure~\ref{sfig:threecat}, bottom-left).
    \item $(\Delta_i=1,D_i=0)$: patients in this category have $T_i \leq C_i$ and $C^D_i > C_i$. Thus, it must be that $T_i \leq C_i < C^D_i=\infty$, and the patient is uncured (Figure~\ref{sfig:threecat}, bottom-right).
    \item $(\Delta_i=1,D_i=1)$: patients in this category should have both $T_i \leq C_i$ and $C^D_i \leq C_i$. However, because $C_i<\infty$, this would imply that both $T_i$ and $C^D_i$ are finite, which is not possible. Therefore, this category cannot occur.
\end{enumerate}
Since the event $(\Delta_i=1, D_i=1)$ is impossible, observing $\Delta_i=1$ uniquely identifies the category $(\Delta_i=1, D_i=0)$. For brevity, we may simply denote this condition by $\Delta_i=1$ throughout the text.

\begin{figure}[h]
\centering
\begin{tikzpicture}[
    mainbox/.style={
        draw=mygrey,
        line width=1pt,
        rounded corners=10pt
    },
    uncuredtag/.style={
        draw=myred!50,
        fill=myred!15,
        text=myred,
        rounded rectangle,
        line width=0.8pt,
        inner xsep=10pt,
        inner ysep=4pt,
        font=\sffamily\footnotesize\bfseries
    },
    curedtag/.style={
        draw=myblue!50,
        fill=myblue!15,
        text=myblue,
        rounded rectangle,
        line width=0.8pt,
        inner xsep=10pt,
        inner ysep=4pt,
        font=\sffamily\footnotesize\bfseries
    },
    timepoint/.style={
        circle,
        fill=black,
        inner sep=1.5pt
    }
]

\begin{scope}[shift={(0,0)}]
    \draw[mainbox] (0,0) rectangle (12.8, 3.0);
    
    \node[draw=mygrey, fill=mygrey!15, text=black, rounded corners=4pt, inner xsep=8pt, inner ysep=4pt, font=\sffamily, anchor=west] at (0.3, 3.0) {
        \tikz\node[circle, fill=black, text=mygrey!15, inner sep=1pt, font=\sffamily\tiny\bfseries]{1}; \hspace{2pt} $\Delta_i = 0, \; D_i = 0$
    };
    
    \begin{scope}[shift={(0.8, 1.1)}]
        \node[uncuredtag] at (2.1, 1.0) {Uncured};
        
        \draw[->, line width=1.1pt] (0,0) -- (4.2,0);
        \draw[line width=1.1pt] (4.2, 0.15) -- (4.2, -0.15);
        \node[above=2pt, font=\sffamily\small] at (4.2, 0.1) {$\infty$};
        
        \node[timepoint] at (1,0) {};
        \node[font=\sffamily] at (1, -0.5) {$C_i$};
        
        \node[timepoint] at (2.2,0) {};
        \node[font=\sffamily] at (2.2, -0.5) {$T_i$};
        
        \node[font=\sffamily] at (4.2, -0.5) {$C^D_i$};
    \end{scope}

    \draw[mygrey!30, line width=0.8pt] (6.4, 0.5) -- (6.4, 2.2);
    \node[fill=white, inner sep=5pt, font=\sffamily\itshape, text=gray!70] at (6.4, 1.35) {or};

    \begin{scope}[shift={(7.5, 1.1)}]
        \node[curedtag] at (2.1, 1.0) {Cured};
        
        \draw[->, line width=1.1pt] (0,0) -- (4.2,0);
        \draw[line width=1.1pt] (4.2, 0.15) -- (4.2, -0.15);
        \node[above=2pt, font=\sffamily\small] at (4.2, 0.1) {$\infty$};
        
        \node[timepoint] at (1,0) {};
        \node[font=\sffamily] at (1, -0.5) {$C_i$};
        
        \node[timepoint] at (2.2,0) {};
        \node[font=\sffamily] at (2.2, -0.5) {$C^D_i$};
        
        \node[font=\sffamily] at (4.2, -0.5) {$T_i$};
    \end{scope}
\end{scope}

\begin{scope}[shift={(0,-4.1)}]
    \draw[mainbox] (0,0) rectangle (5.8, 3.0);
    
    \node[draw=mygrey, fill=mygrey!15, text=black, rounded corners=4pt, inner xsep=8pt, inner ysep=4pt, font=\sffamily, anchor=west] at (0.3, 3.0) {
        \tikz\node[circle, fill=black, text=mygrey!15, inner sep=1pt, font=\sffamily\tiny\bfseries]{2}; \hspace{2pt} $\Delta_i = 0, \; D_i = 1$
    };
    
    \begin{scope}[shift={(0.6, 1.1)}]
        \node[curedtag] at (2.1, 1.0) {Cured};
        
        \draw[->, line width=1.1pt] (0,0) -- (4.2,0);
        \draw[line width=1.1pt] (4.2, 0.15) -- (4.2, -0.15);
        \node[above=2pt, font=\sffamily\small] at (4.2, 0.1) {$\infty$};
        
        \node[timepoint] at (1,0) {};
        \node[font=\sffamily] at (1, -0.5) {$C^D_i$};
        
        \node[timepoint] at (2.2,0) {};
        \node[font=\sffamily] at (2.2, -0.5) {$C_i$};
        
        \node[font=\sffamily] at (4.2, -0.5) {$T_i$};
    \end{scope}
\end{scope}

\begin{scope}[shift={(7.0,-4.1)}]
    \draw[mainbox] (0,0) rectangle (5.8, 3.0);
    
    \node[draw=mygrey, fill=mygrey!15, text=black, rounded corners=4pt, inner xsep=8pt, inner ysep=4pt, font=\sffamily, anchor=west] at (0.3, 3.0) {
        \tikz\node[circle, fill=black, text=mygrey!15, inner sep=1pt, font=\sffamily\tiny\bfseries]{3}; \hspace{2pt} $\Delta_i = 1, \; D_i = 0$
    };
    
    \begin{scope}[shift={(0.6, 1.1)}]
        \node[uncuredtag] at (2.1, 1.0) {Uncured};
        
        \draw[->, line width=1.1pt] (0,0) -- (4.2,0);
        \draw[line width=1.1pt] (4.2, 0.15) -- (4.2, -0.15);
        \node[above=2pt, font=\sffamily\small] at (4.2, 0.1) {$\infty$};
        
        \node[timepoint] at (1.0,0) {};
        \node[font=\sffamily] at (1.0, -0.5) {$T_i$};
        
        \node[timepoint] at (2.2,0) {};
        \node[font=\sffamily] at (2.2, -0.5) {$C_i$};
        
        \node[font=\sffamily] at (4.2, -0.5) {$C^D_i$};
    \end{scope}
\end{scope}

\end{tikzpicture}
\caption{Illustration of the three categories obtained from $\Delta_i$ and $D_i$.}
\label{sfig:threecat}
\end{figure}

\section{Details on the simulation study} \label{app:simulations}

In the simulation study, we run our Bayesian model for 2 chains and 3000 iterations, discarding the first 1500 for each chain. We assess model convergence on each of the 100 simulated datasets by evaluating the effective sample size (ESS) and the $\hat{\mathrm{R}}$. We consider convergence to be reached when ESS is greater than 100 times the number of chains ($100 \times 2=200$ in our case) and $\hat{\mathrm{R}}<1.01$. When one of these conditions is not met, we run the model again with different initial values until convergence is reached. Table~\ref{stab:simres} reports simulation results for the 4 scenarios considered, complementing Figure~\ref{fig:simres} of the main text. 

\begin{table}[!h]
\centering
\caption{Simulation results for the four scenarios considered. True value, bias, and empirical standard error of the point estimate of the estimands $\bar{\delta}$ and $\tau_{\mathrm{UU}}^{\mathrm{RMST}}(t^*=30)$ and the marginal strata probabilities $\bar{\pi}_{g}$ across 100 replications for the two methods. In our Bayesian model, we take the posterior median as the point estimate.}
\label{stab:simres}
\centering
\resizebox{\ifdim\width>\linewidth\linewidth\else\width\fi}{!}{
\begin{tabular}[t]{llcccccccccccc}
\toprule
\multicolumn{2}{c}{ } & \multicolumn{3}{c}{Scenario 1} & \multicolumn{3}{c}{Scenario 2} & \multicolumn{3}{c}{Scenario 3} & \multicolumn{3}{c}{Scenario 4} \\
\cmidrule(l{3pt}r{3pt}){3-5} \cmidrule(l{3pt}r{3pt}){6-8} \cmidrule(l{3pt}r{3pt}){9-11} \cmidrule(l{3pt}r{3pt}){12-14}
Quantity & Method & Truth & Bias & Emp. SE & Truth & Bias & Emp. SE & Truth & Bias & Emp. SE & Truth & Bias & Emp. SE\\
\midrule
$\bar\delta$ & Wang & 0.060 & 0.001 & 0.029 & 0.060 & 0.008 & 0.087 & 0.060 & 0.000 & 0.029 & 0.060 & 0.000 & 0.029\\
 & Ours & 0.060 & 0.003 & 0.029 & 0.060 & -0.002 & 0.082 & 0.060 & 0.003 & 0.028 & 0.060 & 0.000 & 0.029\\
\addlinespace
$\bar\pi_{\mathrm{CC}}$ & Wang & 0.431 & -0.002 & 0.019 & 0.431 & -0.004 & 0.049 & 0.431 & -0.002 & 0.018 & 0.350 & 0.079 & 0.015\\
 & Ours & 0.431 & -0.001 & 0.017 & 0.431 & -0.020 & 0.040 & 0.431 & -0.001 & 0.017 & 0.350 & 0.080 & 0.014\\
\addlinespace
$\bar\pi_{\mathrm{CU}}$ & Wang & 0.129 & 0.003 & 0.019 & 0.129 & 0.012 & 0.061 & 0.129 & 0.002 & 0.019 & 0.211 & -0.080 & 0.020\\
 & Ours & 0.129 & 0.002 & 0.020 & 0.129 & 0.011 & 0.057 & 0.129 & 0.002 & 0.019 & 0.211 & -0.082 & 0.020\\
\addlinespace
$\bar\pi_{\mathrm{UC}}$ & Wang & 0.069 & 0.002 & 0.016 & 0.069 & 0.004 & 0.046 & 0.069 & 0.002 & 0.016 & 0.150 & -0.080 & 0.016\\
 & Ours & 0.069 & -0.001 & 0.015 & 0.069 & 0.012 & 0.040 & 0.069 & -0.001 & 0.015 & 0.150 & -0.082 & 0.016\\
\addlinespace
$\bar\pi_{\mathrm{UU}}$ & Wang & 0.371 & -0.003 & 0.019 & 0.371 & -0.013 & 0.056 & 0.371 & -0.002 & 0.019 & 0.289 & 0.081 & 0.019\\
 & Ours & 0.371 & -0.001 & 0.019 & 0.371 & -0.008 & 0.052 & 0.371 & -0.001 & 0.019 & 0.289 & 0.084 & 0.018\\
\addlinespace
$\tau_{\mathrm{UU}}^{\mathrm{RMST}}(t^*=30)$ & Wang & 0.945 & -0.103 & 0.297 & 0.945 & -0.008 & 0.917 & 0.585 & -0.324 & 0.221 & 0.918 & -0.305 & 0.219\\
 & Ours & 0.945 & -0.116 & 0.305 & 0.945 & -0.050 & 0.814 & 0.585 & -0.132 & 0.220 & 0.918 & -0.312 & 0.250\\
\bottomrule
\end{tabular}}
\end{table}

\section{Details on the NIVAS study and additional analyses}
\label{app:nivas}

In this section we provide additional information on the baseline covariates and missing data handling for the NIVAS study. Moreover, we report further results from the data analysis presented in the main text, as well as results from additional models we fit.

\subsection{Baseline covariates and missing data handling}

Table~\ref{stab:covariates} reports descriptive summaries of the 12 baseline covariates considered in the analysis, as well as the number of missing values by treatment group. We address missing data via multiple imputation using the \texttt{mice} package in R \citep{mice}, and we follow the guidelines in \citet{zhou2010note} for Bayesian inference after multiple imputation. Specifically, we generate $M=10$ imputed datasets and obtain posterior samples of the quantities of interest. We then combine the $M$ sets of posterior draws, obtaining the final posterior distributions used for inference. 

\begin{table}[!h]
\centering
\caption{Summary statistics and number of missing values of baseline covariates by treatment group: NIV ($n=148$) and Oxygen ($n=145$).}
\label{stab:covariates}
\centering
\resizebox{\ifdim\width>\linewidth\linewidth\else\width\fi}{!}{
\begin{tabular}[t]{lcccc}
\toprule
\multicolumn{1}{c}{ } & \multicolumn{2}{c}{Summary statistics} & \multicolumn{2}{c}{Missing values (no.)} \\
\cmidrule(l{3pt}r{3pt}){2-3} \cmidrule(l{3pt}r{3pt}){4-5}
Variables & NIV & Oxygen & NIV & Oxygen\\
\midrule
\textit{Continuous -- mean (sd)} &  &  &  & \\
pH & 7.42 (0.07) & 7.41 (0.07) & 14 & 21\\
PaO$_2$/FiO$_2$ (mm Hg) & 200.80 (69.03) & 187.85 (70.96) & 14 & 20\\
\addlinespace
\textit{Binary -- no. (\%)} &  &  &  & \\
Male gender & 116 (78\%) & 108 (74\%) & 0 & 0\\
IGS score $>$ 40 & 34 (24\%) & 31 (22\%) & 4 & 5\\
Psychotropic use & 15 (10\%) & 16 (11\%) & 1 & 1\\
White cell count $>$ 20000 n/$\mu$liter & 17 (14\%) & 17 (14\%) & 26 & 27\\
Oesophagectomy & 14 (10\%) & 9 (6\%) & 8 & 2\\
Epidural analgesia & 23 (16\%) & 21 (14\%) & 0 & 0\\
Extubated $<$ 6-hr after the end of surgery & 89 (65\%) & 88 (63\%) & 11 & 6\\
Copious tracheal secretions & 58 (40\%) & 54 (38\%) & 3 & 2\\
Age $\geq 60$ & 92 (62\%) & 89 (61\%) & 0 & 0\\
Upper abdominal surgery & 93 (63\%) & 91 (63\%) & 0 & 0\\
\bottomrule
\end{tabular}}
\end{table}

\subsection{Additional results}

We now present additional data analysis results, comparing different models we fit. To facilitate presentation, we refer to the models using the labels summarized in Table~\ref{stab:model-labels}. 

Model 1, presented in the main text, assumes monotonicity (formulated to rule out $\mathrm{UC}$ patients) and incorporates auxiliary information on the cure status for censored patients, represented by the binary variable $D_i$. In our study, patients are considered cured if they are discharged from the hospital without experiencing the outcome. To account for this, censoring is applied at hospital discharge, loss to follow-up, or 30 days, whichever occurs first. Specifically, if a patient is censored due to hospital discharge, we set $D_i = 1$ (indicating they are cured). Conversely, if censoring occurs due to loss to follow-up or at 30 days, we set $D_i = 0$ (meaning their cure status remains unknown). 
Model 2 shares the same specification as Model 1 but relaxes the monotonicity assumption. Comparing these two models allows us to assess the sensitivity of our results to this assumption. 
Model 3, instead, maintains the monotonicity assumption but does not include the auxiliary information on the cure status. Consequently, patients are censored only at loss to follow-up or 30 days, whichever occurs first. 
Finally, Model 4 replicates Model 3 but relaxes the monotonicity assumption.

\begin{table}[h]
    \centering
    \caption{Model labels and corresponding characteristics. Monotonicity indicates whether the monotonicity assumption (formulated to rule out $\mathrm{UC}$ patients) is assumed. Cure information indicates whether additional information on the patient's cure status is included in the likelihood. Censoring specifies when an observation is censored (whichever of those events comes first).}
    \label{stab:model-labels}
    \resizebox{\ifdim\width>\linewidth\linewidth\else\width\fi}{!}{
    \begin{tabular}{llll}
        \toprule
        Model label & Monotonicity & Cure information $D_i$ & Censoring \\
        \midrule
        Model 1 & Yes & Included & Hospital discharge, loss to follow-up, or 30 days \\
        Model 2 & No & Included & Hospital discharge, loss to follow-up, or 30 days \\
        Model 3 & Yes & Not included & Loss to follow-up or 30 days \\
        Model 4 & No & Not included & Loss to follow-up or 30 days \\
        \bottomrule
    \end{tabular}
    }
\end{table}

\paragraph{Model 1}
Figure~\ref{sfig:res_cure} shows the marginal survival functions estimated with Model 1, as a complement to what is presented in the main text. We report both the marginal survival functions by principal stratum, as well as by treatment group. The latter agree with the corresponding Kaplan-Meier estimators, showing that our estimates are consistent with the observed data. Table~\ref{stab:res_cure} reports the same posterior summaries presented in the main text. 

\begin{table}[!h]
\centering
\caption{Model 1. Posterior median and 95\% credible intervals of selected quantities.}
\label{stab:res_cure}
\centering
\begin{tabular}[t]{lcc}
\toprule
Quantity & Median & 95\% credible interval\\
\midrule
$\bar\delta$ & 0.03 & (0.01, 0.10)\\
$\tau_{\mathrm{UU}}^{\mathrm{RMST}}(t^*=30)$ & 1.27 & (-0.76, 3.82)\\
$\tau_{\neg\mathrm{CC}}^{\mathrm{RMST}}(t^*=30)$ & 2.80 & (0.31, 6.16)\\
$\bar\pi_{\mathrm{CC}}$ & 0.53 & (0.47, 0.58)\\
$\bar\pi_{\mathrm{CU}}$ & 0.03 & (0.01, 0.10)\\
$\bar\pi_{\mathrm{UU}}$ & 0.44 & (0.38, 0.49)\\
\bottomrule
\end{tabular}
\end{table}

\paragraph{Model 2}
We fit this model to assess sensitivity to the monotonicity assumption, which is now relaxed so that $\mathrm{UC}$ patients are allowed. Auxiliary information on the cure status for censored patients is still included. Posterior summaries of the relevant quantities are reported in Figure~\ref{sfig:res_cure_nomonot} and Table~\ref{stab:res_cure_nomonot}. 
The effect for always-uncured patients is quite robust to the monotonicity assumption, as we do not find substantial changes in the the posterior of $\tau_{\mathrm{UU}}^{\mathrm{RMST}}(t^*=30)$ and $\tau_{\mathrm{UU}}(t)$ compared to Model 1. 
Instead, the effect on the cure probability ($\bar\delta$) and on the survival for the non-always-cured union of strata ($\tau_{\neg\mathrm{CC}}^{\mathrm{RMST}}(t^*=30)$ and $\tau_{\neg\mathrm{CC}}(t)$) are now attenuated. This behavior is coherent with the drop of the monotonicity assumption, since we now include an additional stratum that directly impacts these estimands. Nonetheless, there is a moderate-to-high posterior probability of a positive effect for both $\bar\delta$ (73\%) and $\tau_{\neg\mathrm{CC}}^{\mathrm{RMST}}(t^*=30)$ (88\%).
No relevant change in the principal stratum proportions is detected, except for some adjustments due to the addition of the $\mathrm{UC}$ stratum. Most of the patients are still always-cured or always-uncured. The $\mathrm{CU}$ stratum remains relatively small, and the $\mathrm{UC}$ stratum is even smaller. 
Overall, we find that removing monotonicity does not substantially change our conclusions.  

\begin{table}[!h]
\centering
\caption{Model 2. Posterior median and 95\% credible intervals of selected quantities.}
\label{stab:res_cure_nomonot}
\centering
\begin{tabular}[t]{lcc}
\toprule
Quantity & Median & 95\% credible interval\\
\midrule
$\bar\delta$ & 0.01 & (-0.04, 0.09)\\
$\tau_{\mathrm{UU}}^{\mathrm{RMST}}(t^*=30)$ & 1.27 & (-0.84, 4.13)\\
$\tau_{\neg\mathrm{CC}}^{\mathrm{RMST}}(t^*=30)$ & 1.92 & (-1.32, 5.68)\\
$\bar\pi_{\mathrm{CC}}$ & 0.51 & (0.44, 0.57)\\
$\bar\pi_{\mathrm{CU}}$ & 0.04 & (0.01, 0.12)\\
$\bar\pi_{\mathrm{UC}}$ & 0.02 & (0.00, 0.08)\\
$\bar\pi_{\mathrm{UU}}$ & 0.42 & (0.36, 0.48)\\
\bottomrule
\end{tabular}
\end{table}

\paragraph{Model 3}
This specification replicates Model 1 but does not include auxiliary information on the cure status for censored patients. Comparing posterior distributions (Figure~\ref{sfig:res} and Table~\ref{stab:res}), we find that posterior medians are quite similar to Model 1, except for $\bar{\delta}$, $\tau_{\neg\mathrm{CC}}^{\mathrm{RMST}}(t^*=30)$, and $\bar\pi_{\mathrm{CU}}$ showing a bigger change. In terms of credible intervals however, these are always narrower in Model 1, highlighting how incorporating auxiliary information on the cure status, when available, sharpens posterior inference.

\begin{table}[!h]
\centering
\caption{Model 3. Posterior median and 95\% credible intervals of selected quantities.}
\label{stab:res}
\centering
\begin{tabular}[t]{lcc}
\toprule
Quantity & Median & 95\% credible interval\\
\midrule
$\bar\delta$ & 0.08 & (0.02, 0.21)\\
$\tau_{\mathrm{UU}}^{\mathrm{RMST}}(t^*=30)$ & 1.42 & (-1.94, 6.27)\\
$\tau_{\neg\mathrm{CC}}^{\mathrm{RMST}}(t^*=30)$ & 4.24 & (0.75, 8.37)\\
$\bar\pi_{\mathrm{CC}}$ & 0.48 & (0.37, 0.55)\\
$\bar\pi_{\mathrm{CU}}$ & 0.08 & (0.02, 0.21)\\
$\bar\pi_{\mathrm{UU}}$ & 0.43 & (0.37, 0.51)\\
\bottomrule
\end{tabular}
\end{table}

\paragraph{Model 4}
Finally, we estimate Model 4, which replicates the specification of Model 3 but relaxes the monotonicity assumption. Results are shown in Figure~\ref{sfig:res_nomonot} and Table~\ref{stab:res_nomonot}. A comparison of the posterior distributions between Model 3 and Model 4 yields findings that mirror those obtained when contrasting Model 1 and Model 2 in terms of sensitivity to monotonicity. Furthermore, by comparing the posteriors of Model 4 with those of Model 2, neither of which imposes monotonicity, we reaffirm that incorporating auxiliary information on the cure status enhances estimation precision, leading to narrower credible intervals in our empirical setting.

\begin{table}[h!]
\centering
\caption{Model 4. Posterior median and 95\% credible intervals of selected quantities.}
\label{stab:res_nomonot}
\centering
\begin{tabular}[t]{lcc}
\toprule
Quantity & Median & 95\% credible interval\\
\midrule
$\bar\delta$ & 0.04 & (-0.13, 0.20)\\
$\tau_{\mathrm{UU}}^{\mathrm{RMST}}(t^*=30)$ & 1.72 & (-1.87, 7.64)\\
$\tau_{\neg\mathrm{CC}}^{\mathrm{RMST}}(t^*=30)$ & 2.72 & (-1.23, 6.90)\\
$\bar\pi_{\mathrm{CC}}$ & 0.40 & (0.21, 0.51)\\
$\bar\pi_{\mathrm{CU}}$ & 0.12 & (0.04, 0.28)\\
$\bar\pi_{\mathrm{UC}}$ & 0.08 & (0.01, 0.24)\\
$\bar\pi_{\mathrm{UU}}$ & 0.39 & (0.30, 0.49)\\
\bottomrule
\end{tabular}
\end{table}


\begin{figure}
    \centering
    \includegraphics[width=0.9\textwidth]{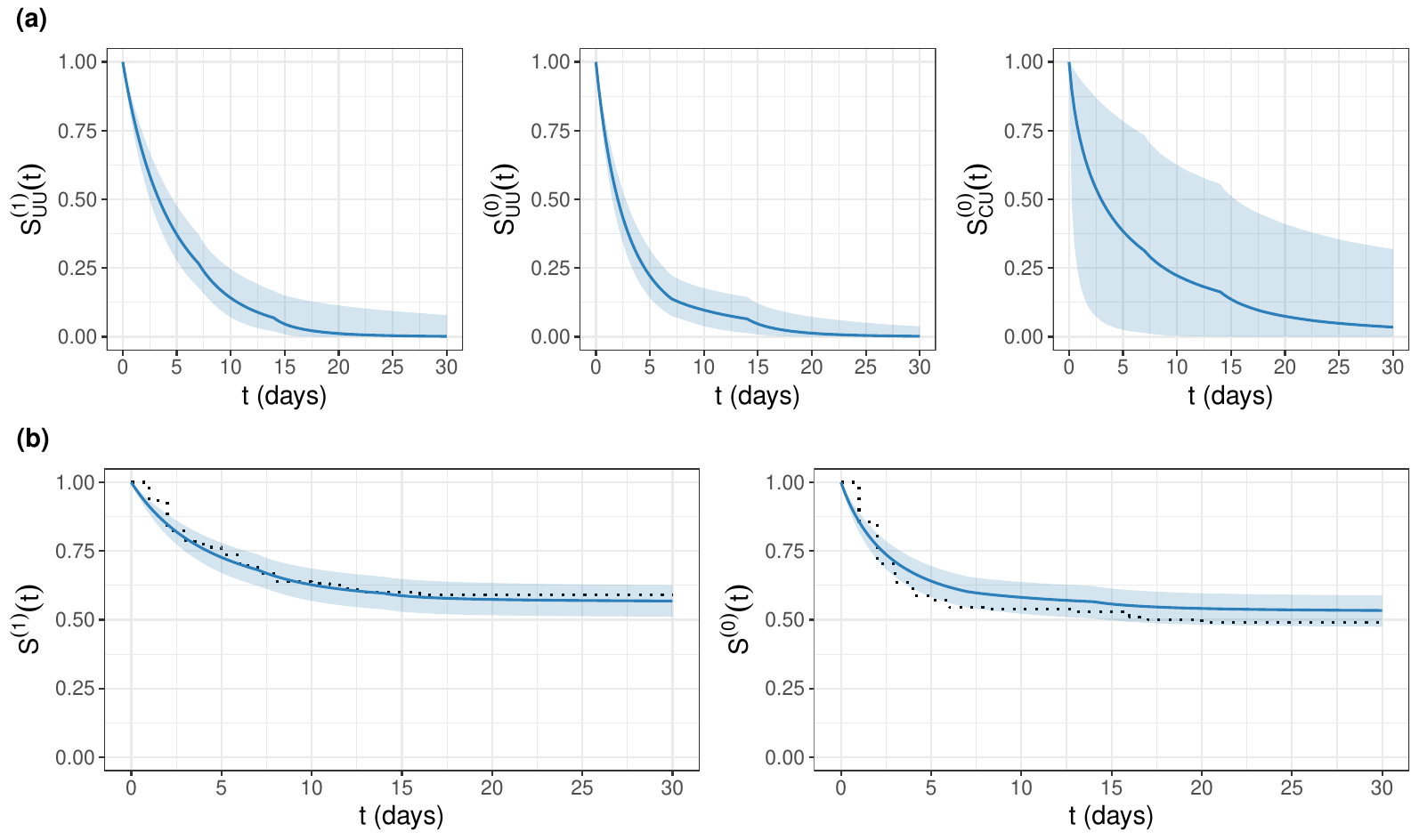}
    \caption{Model 1. Posterior summaries of the marginal survival curves. Panel (a): posterior median (solid line) and 95\% credible intervals (shaded area) of marginal survival curves by principal stratum. Panel (b): posterior median (solid line) and 95\% credible intervals (shaded area) of marginal survival curves by treatment group, where dotted lines show the Kaplan-Meier estimator of the corresponding survival curve.}
    \label{sfig:res_cure}
\end{figure}

\begin{figure}
    \centering
    \includegraphics[width=0.8\linewidth]{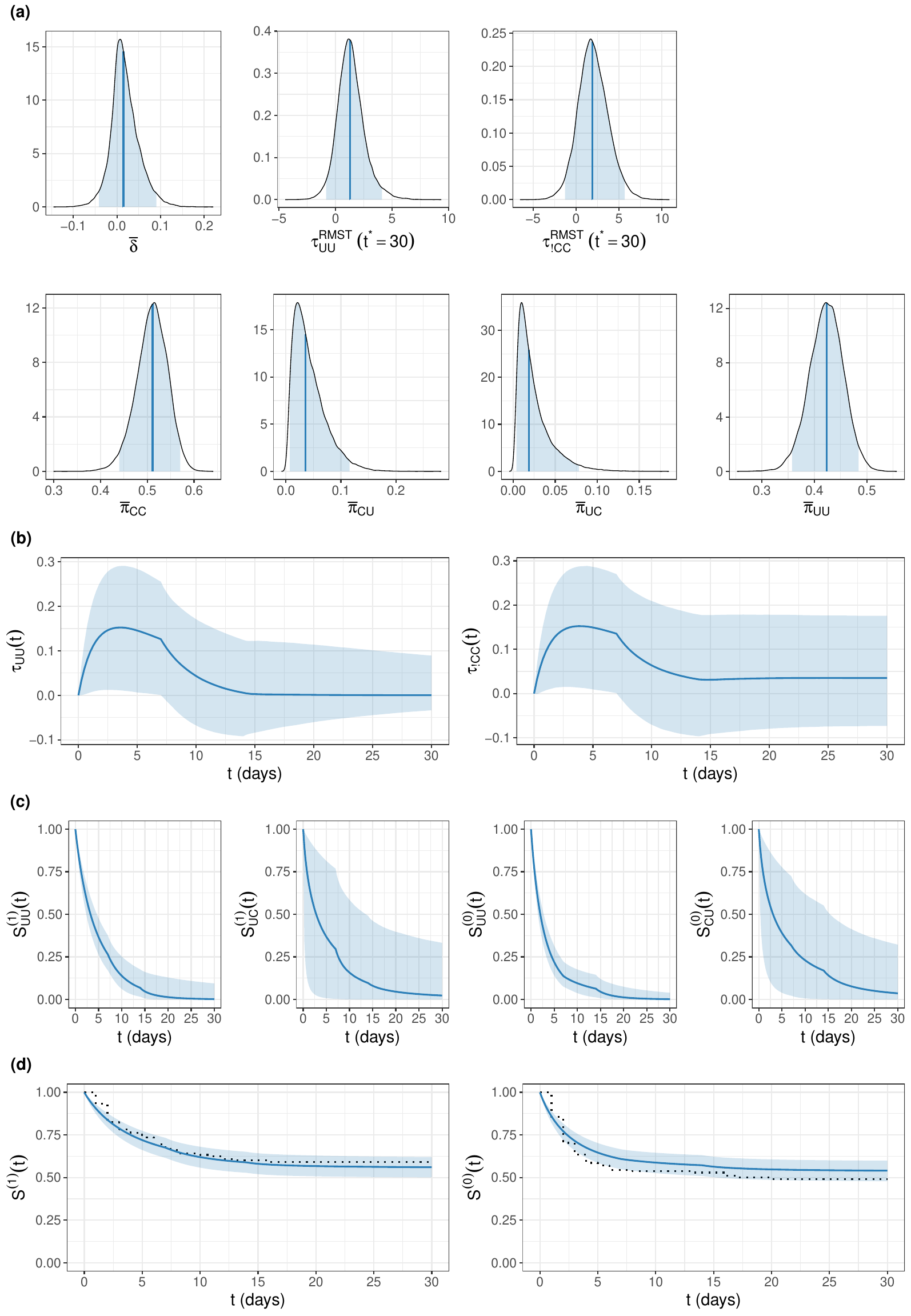}
    \caption{Model 2. Posterior summaries of the relevant quantities. 
  Panel (a): posterior distributions of the difference in cure probabilities, difference in RMST for $\mathrm{UU}$ and $\neg\mathrm{CC}$, and marginal principal strata probabilities. The vertical solid line indicates the posterior median, and the shaded area represents the 95\% credible interval. 
  Panel (b): posterior median (solid line) and 95\% credible intervals (shaded area) of the difference in survival probability for $\mathrm{UU}$ and $\neg\mathrm{CC}$ patients over time.
  Panel (c): posterior median (solid line) and 95\% credible intervals (shaded area) of marginal survival curves by principal stratum. 
  Panel (d): posterior median (solid line) and 95\% credible intervals (shaded area) of marginal survival curves by treatment group, where dotted lines show the Kaplan-Meier estimator of the corresponding survival curve.}
    \label{sfig:res_cure_nomonot}
\end{figure}

\begin{figure}
    \centering
    \includegraphics[width=0.8\linewidth]{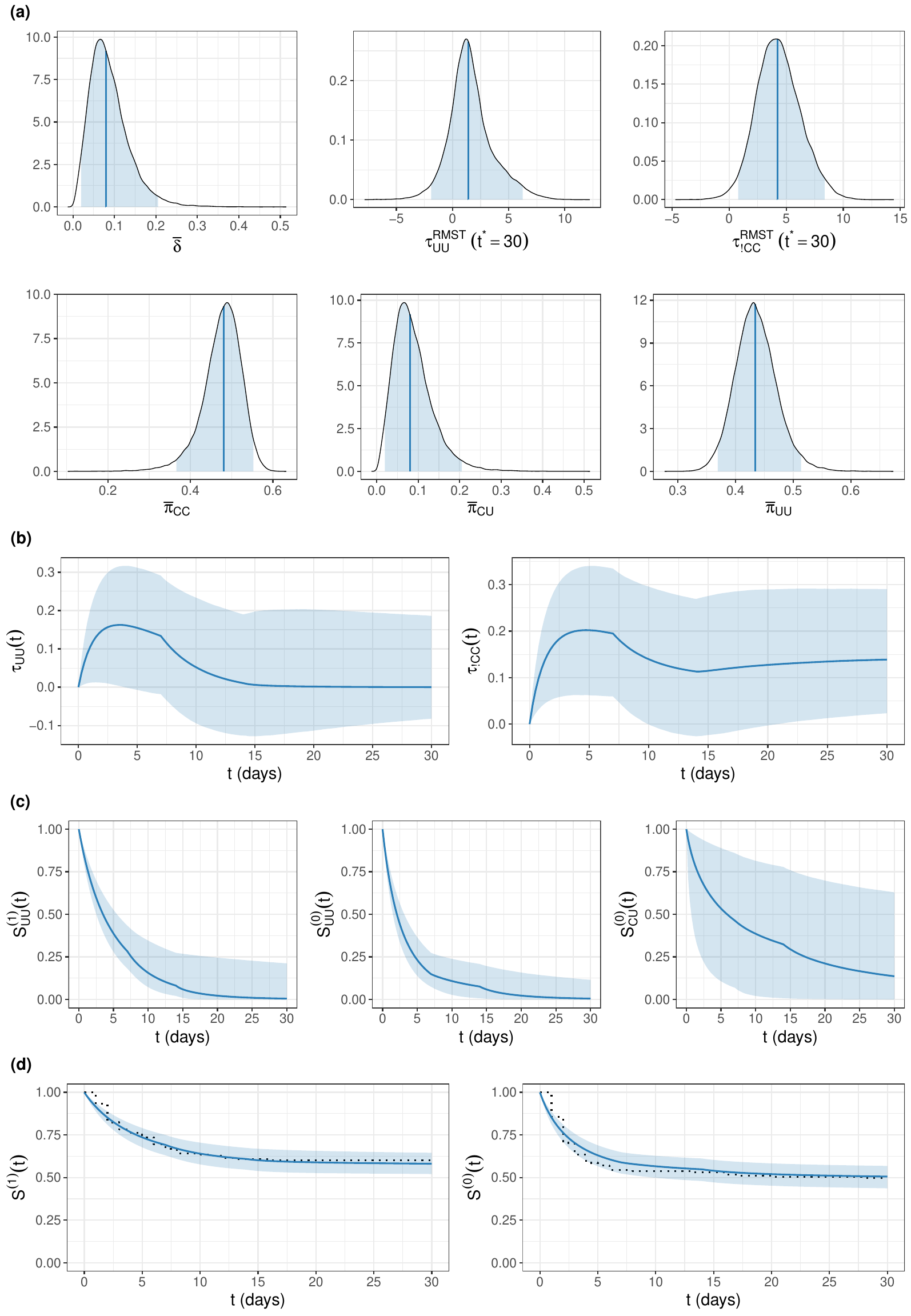}
    \caption{Model 3. Posterior summaries of the relevant quantities. 
  Panel (a): posterior distributions of the difference in cure probabilities, difference in RMST for $\mathrm{UU}$ and $\neg\mathrm{CC}$, and marginal principal strata probabilities. The vertical solid line indicates the posterior median, and the shaded area represents the 95\% credible interval. 
  Panel (b): posterior median (solid line) and 95\% credible intervals (shaded area) of the difference in survival probability for $\mathrm{UU}$ and $\neg\mathrm{CC}$ patients over time.
  Panel (c): posterior median (solid line) and 95\% credible intervals (shaded area) of marginal survival curves by principal stratum. 
  Panel (d): posterior median (solid line) and 95\% credible intervals (shaded area) of marginal survival curves by treatment group, where dotted lines show the Kaplan-Meier estimator of the corresponding survival curve.}
    \label{sfig:res}
\end{figure}

\begin{figure}
    \centering
    \includegraphics[width=0.8\linewidth]{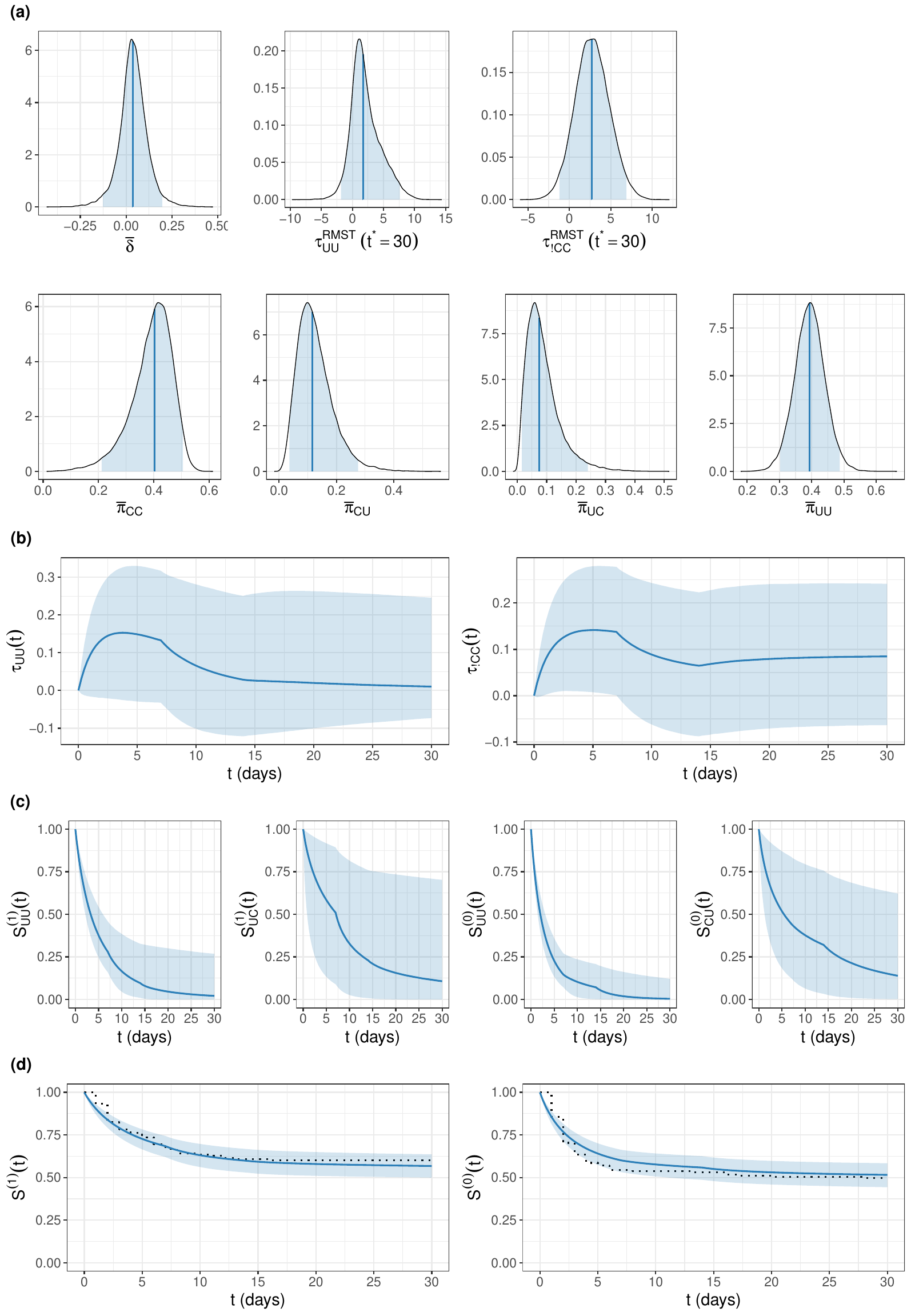}
    \caption{Model 4. Posterior summaries of the relevant quantities. 
  Panel (a): posterior distributions of the difference in cure probabilities, difference in RMST for $\mathrm{UU}$ and $\neg\mathrm{CC}$, and marginal principal strata probabilities. The vertical solid line indicates the posterior median, and the shaded area represents the 95\% credible interval. 
  Panel (b): posterior median (solid line) and 95\% credible intervals (shaded area) of the difference in survival probability for $\mathrm{UU}$ and $\neg\mathrm{CC}$ patients over time.
  Panel (c): posterior median (solid line) and 95\% credible intervals (shaded area) of marginal survival curves by principal stratum. 
  Panel (d): posterior median (solid line) and 95\% credible intervals (shaded area) of marginal survival curves by treatment group, where dotted lines show the Kaplan-Meier estimator of the corresponding survival curve.}
    \label{sfig:res_nomonot}
\end{figure}

\end{document}